\documentclass{article}%
\usepackage{amssymb}
\usepackage{amsmath}
\usepackage{amsfonts}
\usepackage{graphicx}%
\providecommand{\U}[1]{\protect\rule{.1in}{.1in}}
\newtheorem{theorem}{Theorem}

\newtheorem{lemma}[theorem]{Lemma}

\newenvironment{proof}[1][Proof]{\noindent\textbf{#1.} }{\ \rule{0.5em}{0.5em}}
\begin{document}
\vskip3mm

\noindent\textbf{ON THE FIRST ORDER ASYMPTOTIC THEORY OF QUANTUM ESTIMATION} \vskip3mm

\vskip3mm \noindent K. Matsumoto\quad\quad\quad

\noindent Quantum Computation Group, National Institute of Informatics

\noindent2-1-2, Hitotsubashi, Chiyoda-ku, Tokyo 101-8430

\noindent keiji@nii.ac.jp

 \vskip3mm

\noindent ABSTRACT

We give a rigorous treatment on the foundation of the first order asymptotic
theory of quantum estimation, with tractable and reasonable regularity
conditions. Different from past works, we do not use Fisher information nor
MLE, and an optimal estimator is constructed based on locally unbiased
estimators. Also, we treat state estimation by local operations and classical
communications (LOCC), and estimation of quantum operations.

\section*{\noindent{\protect\normalsize 1. INTRODUCTION}}

\vskip -3mm The purpose of this paper is to give a rigorous foundation of the
first order asymptotic theory of quantum estimation, which has been
established in these years. In addition to most basic setting, we also treat
state estimation by local operations and classical communications (LOCC, in
short) and estimation of quantum operations.

This research field was initiated by Nagaoka\thinspace(1987),
Nagaoka\thinspace(1989), followed by Hayashi and Matsumoto\thinspace(1998),
Gill and Massar\thinspace(2002). (Many of important papers in the field are
included in Hayashi\thinspace(2005).) Relying on classical estimation theory,
especially the fact that the inverse of Fisher information gives the optimal
efficiency of consistent estimators, they had reduced the optimization of
consistent estimators to optimization of Fisher information, or equivalently,
of locally unbiased estimators. These works had laid foundation on which
number of works, mostly computation of asymptotically optimal estimators and
their costs, are based. In closer look, however, they either miss the detail
of the proof, or assume intractable regularity conditions.

One reason for such incompleteness is that the focus of these works were
consequences of the foundations, rather than their rigorous proof. Also, the
following technical difficulties seems to be a part of reasons. In quantum
statistics, the probability distribution of the data depends on the choice of
measurement. Therefore, for classical estimation theory to be applicable, a
set of regularity conditions should hold for \textit{all} the probability
distributions resulting from arbitrary measurement of interest. In Hayashi and
Matsumoto\thinspace(1998), they use this sort of statement as their regularity
condition. As a result, their regularity conditions are quite difficult to
check for given quantum statistical models.

The purpose of the paper is to provide rigorous proof assuming tractable
regularity conditions, including the case of infinite dimensional Hilbert
space. In addition to the most basic settings, we also treat state estimation
by semi-classical measurement and by local operations and classical
communications (LOCC, in short). Also, estimation of quantum operations is studied.

Different from previous works, we avoided use of Fisher information, and
composed an asymptotically efficient estimator from an optimal locally
unbiased estimator, because of the following reasons. First, quantum
asymptotic Crammer-Rao bound is not a simple function of any quantum analogue
of Fisher information. It equals Holevo bound, which is defined in terms of
operator version of asymptotically unbiasedness conditions (Hayashi and
Matsumoto\thinspace(2004), Matsumoto\thinspace(1999), Guta and
Jencova\thinspace(2006)). The second motivation is to simplify the regularity
conditions, by avoiding technical difficulties stated above.

One of major difference between quantum mechanics and classical mechanics is
behavior of composite systems. In quantum mechanics, the state of the system
and the measurement in composite systems may not be in convex combinations of
those without correlations between subsystems. In such cases, we often observe
non-trivial quantum effects, which can \textit{never} be reproduced by
classical mechanical random variables, such as violation of Bell's inequality.
Therefore, it is of interest to compare measurement with non-trivial
correlations and the one without it in their efficiency of state estimation.

For that purpose, we study \textit{semi-classical} measurements and
\textit{LOCC} (, short for local operations and classical communications,)
measurements. In the former, we are not allowed to use measurement
collectively acts on given $n$ independent samples. In the latter, each sample
is a state in a composite system (A and B, say), and we are not allowed to use
the measurement quantumly correlating over A-B split.

The last topic is estimation of a quantum operations. It had been observed
that for some cases (e.g., unitary operations, or noiseless operations), the
mean square error of optimal estimators scales as $O\left(  1/n^{2}\right)  $
(Heisenberg rate), which is significantly smaller than $O\left(  1/n\right)
$, and there had been suggestion of efficient measurement scheme utilizing
this effect. Recently, however, several authors ( Fujiwara\thinspace(2005),
Zhengfeng Ji, et. al.\thinspace(2006), etc) had pointed out that $O\left(
1/n^{2}\right)  $-scaling is not observed in some class of operations
(typically, they corresponds to noisy operations). We show that $O\left(
1/n^{2}\right)  $-scaling is rather exceptional, and not observed so long as
the model lies in interior of the totality of quantum operations.
\vspace{-7mm}

\section*{{\protect\normalsize \noindent2. QUANTUM ESTIMATION THEORY}}

\vskip -5mm

\subsection*{\noindent{\protect\normalsize 2.1. QUANTUM STATE AND
MEASUREMENT}}

\vskip-3mm In quantum mechanics, the probability distribution of data $z\in%
\mathbb{R}
^{l}$ is a function of the \textit{state} $\rho$ of the system of interest,
and the \textit{measurement} $\mathsf{M}$ which is applied to the system. The
probability that $\omega$ lies in a Borel set $\Delta$, the corresponding
random variable, and the post-measurement state is denoted by $P_{\rho
}^{\mathsf{M}}\left(  \Delta\right)  $, $\Omega$, and $\rho_{\Delta
}^{\mathsf{M}}$, respectively. (Throughout the paper, the random variable is
denoted by capital letters, and the elements of its range is denoted by its decapitalization.)

$\rho$ and $\mathsf{M}$ are represented by linear operators defined in a
separable Hilbert space ($\mathcal{H}$, say). The inner product of $\varphi$
and $\psi$ is denoted by $\left\langle \varphi,\,\psi\right\rangle $. We
assign to the composite of the system $\mathcal{H}_{1}$ and \ $\mathcal{H}%
_{2}$ the tensor product $\mathcal{H}_{1}\otimes\mathcal{H}_{2}$, which is the
linear span of $\left\{  e_{1,i}\otimes e_{2,j}\right\}  $ ( $\left\{
e_{1,i}\right\}  $ and $\left\{  e_{2,i}\right\}  $ be a complete orthonormal
basis (CONS) of $\mathcal{H}_{1}$ and \ $\mathcal{H}_{2}$ respectively).

The notation $|A|$ means $|A|:=\left(  AA^{\dagger}\right)  ^{1/2}$, and
$\left\Vert A\right\Vert _{1}:=\mathrm{tr}\,|A|$ is a quantum version of total
variation. The totality of \textit{trace class operators}, or operators with
$\left\Vert A\right\Vert _{1}<\infty$, is denoted by $\mathcal{\tau}c\left(
\mathcal{H}\right)  $. Also, $\left\Vert A\right\Vert :=\sup_{\left\Vert
\varphi\right\Vert =1}\left\Vert A\varphi\right\Vert $ and $\ \mathcal{B}%
\left(  \mathcal{H}\right)  $ denotes the totality of \textit{bounded
operators}, or operators with $\left\Vert A\right\Vert <\infty$. (the standard
norm in $%
\mathbb{R}
^{m}$ and in $\mathcal{H}$ is also denoted by $\left\Vert \cdot\right\Vert $.)
We introduce an order in the space of matrices by $A\geq(>)\,B\Leftrightarrow
\left\langle \varphi,A\varphi\right\rangle \geq(>)\,\left\langle
\varphi,B\varphi\right\rangle $, $\forall\varphi$. An operator $A$ is said to
be \textit{positive}, if $A\geq0$. \ A mapping $\Lambda$ of $\mathcal{\tau
}c\left(  \mathcal{H}\right)  $ to $\mathcal{\tau}c\left(  \mathcal{H}%
^{\prime}\right)  $ is called \textit{completely positive}, if $\Lambda
\otimes\mathbf{I}:\mathcal{B}\left(  \mathcal{H\otimes K}\right)
\rightarrow\mathcal{B}\left(  \mathcal{H}^{\prime}\mathcal{\otimes K}\right)
$ is positive, i.e., $A\geq0\Rightarrow\Lambda\otimes\mathbf{I}\left(
A\right)  \geq0$. $\Lambda$ is said to be \textit{trace preserving} if
$\mathrm{\mathrm{tr}\,}X=\mathrm{\mathrm{tr}\,}\Lambda\left(  X\right)  $
($\forall X$). Also we define $\left\Vert \Lambda\right\Vert _{cb}%
:=\sup_{X:\left\Vert X\right\Vert _{1}=1}\left\Vert \Lambda\otimes
\mathbf{I}\left(  X\right)  \right\Vert _{1}$.

A state of the system is represented by a \textit{density operator}, \ or an
operator $\rho$ with $\rho\geq0$, $\rho=\rho^{\ast}$, and $\mathrm{tr}%
\,\rho=1$. A measurement $\mathsf{M}$ is represented by an \textit{instrument}%
, or a $\sigma$additive map $\mathsf{M}:\Delta\rightarrow\mathsf{M}\left[
\Delta\right]  $ of the collection $\mathfrak{B}$ of Borel subsets in $%
\mathbb{R}
^{m}$ into a completely positive linear transform $\mathsf{M}\left[
\Delta\right]  $ in $\mathcal{\tau}c\left(  \mathcal{H}\right)  $ with
$\,\mathsf{M}\left[
\mathbb{R}
^{l}\right]  $'s being trace-preserving. Here, $\sigma$-additivity is in the
sense of strong operator topology in $\mathcal{B}\left(  \mathcal{\tau
}c\left(  \mathcal{H}\right)  \right)  $. \ Using $\rho$ and $\mathsf{M}%
\left[  \Delta\right]  $, $P_{\rho}^{\mathsf{M}}\left(  \Delta\right)  $ and
$\rho_{\Delta}$ is given by $\mathrm{tr}\,\mathsf{M}\left[  \Delta\right]
\rho$ and $\frac{1}{P_{\rho}^{\mathsf{M}}\left(  \Delta\right)  }%
\mathsf{M}\left[  \Delta\right]  \left(  \rho\right)  $, respectively. An
operation which does not extract information is described by a completely
positive and trace-preserving (CPTP) linear map $\Lambda$ from $\mathcal{\tau
}c\left(  \mathcal{H}\right)  $ to $\mathcal{\tau}c\left(  \mathcal{H}%
^{\prime}\right)  $.\ 

When we are interested only in $P_{\rho}^{\mathsf{M}}\left(  \Delta\right)  $,
we use a \textit{positive operator valued measure} (\textit{POVM}, in short),
or a $\sigma$-additive map $M:\Delta\rightarrow M\left(  \Delta\right)  $ of
$\mathfrak{B}$ to positive Hermitian operators with $M\left(
\mathbb{R}
^{l}\right)  =\mathbf{1}$. Here, $\sigma$-additivity \ is in the sense of weak
operator topology in $\mathcal{B}\left(  \mathcal{H}\right)  $. The POVM $M$
corresponding to the measurement $\mathsf{M}$ satisfy $P_{\rho}^{\mathsf{M}%
}\left(  \Delta\right)  =\mathrm{tr}\,\mathsf{M}\left[  \Delta\right]
\rho=\mathrm{tr}\,\rho M\left(  \Delta\right)  $. Throughout the paper, POVM
of a measurement is denoted by the same character as the measurement but in
the standard font.

The \textit{support} $\mathrm{supp}\left(  \mathsf{M}\right)  $ of the
instrument $\mathsf{M}$ over $\mathfrak{B}\left(
\mathbb{R}
^{l}\right)  $ is the smallest set with $\mathsf{M}\left[  \mathrm{supp}%
\left(  \mathsf{M}\right)  \right]  =\mathsf{M}\left[
\mathbb{R}
^{l}\right]  $. The support of a POVM and a measure over $\mathfrak{B}\left(
\mathbb{R}
^{l}\right)  $ are defined analogously.

In this paper, we need integral of the function taking values in
$\mathcal{\tau}c\left(  \mathcal{H}\right)  $ and $\mathcal{B}\left(
\mathcal{\tau}c\left(  \mathcal{H}\right)  \right)  $, which is a Banach space
with the norm $\left\Vert \cdot\right\Vert _{1}$ and $\left\Vert
\cdot\right\Vert _{cb}$, respectively. A Banach space valued function $f$ is
called strongly measurable iff $\forall\varepsilon>0$ $\exists f^{\prime}$
$\left\Vert f\left(  x\right)  -f^{\prime}\left(  x\right)  \right\Vert
<\varepsilon$ holds almost everywhere. $f$ is called weakly measurable iff
$\left\langle y^{\ast},f\left(  x\right)  \right\rangle $ is measurable for
any element $y^{\ast}$ of the dual space. Since $\mathcal{\tau}c\left(
\mathcal{H}\right)  $ and $\mathcal{B}\left(  \mathcal{\tau}c\left(
\mathcal{H}\right)  \right)  $ are separable, these two concepts are
equivalent in our case due to Theorem\thinspace1.1.4 of Schwabik and
Guoju\thinspace(2005).

Pettis integral of weakly measurable function $f$ is defined by the relation
$\int\left\langle y^{\ast},f\left(  x\right)  \right\rangle \mathrm{d}%
x=\left\langle y^{\ast},\int f\left(  x\right)  \mathrm{d}x\right\rangle $,
$\forall y$. Bochner integral of a simple function $\sum_{i}c_{i}\chi_{A_{i}}$
is defined as $\sum_{i}c_{i}\mu\left(  A_{i}\right)  $. For a strongly
measurable function $f$, it is defined as $\lim_{n\rightarrow\infty}\int
f_{n}\left(  x\right)  \mathrm{d}x$ (convergent in norm), where $\left\{
f_{n}\right\}  _{n}$ is a sequence of simple functions with $\lim
_{n\rightarrow\infty}\left\Vert f_{n}\left(  x\right)  -f\left(  x\right)
\right\Vert =0$ almost everywhere. Bochner integral exists iff $\int\left\Vert
f\right\Vert \mathrm{d}x<\infty$ (Theorem\thinspace1.4.3 of Schwabik and
Guoju\thinspace(2005)). Fubini's theorem holds for Pettis integral and Bochner integral.

\vskip-3mm

\subsection*{{\protect\normalsize \noindent2.2 ASYMPTOTIC\ THEORY OF QUANTUM
STATE\ ESTIMATION}}

\vskip-3mm Suppose that we are given $n$ independently and identically
prepared samples, i.e., the system $\underset{n}{\underbrace{\mathcal{H}%
\otimes\cdots\otimes}}\mathcal{H}:=\mathcal{H}^{\otimes n}$ in the state
$\underset{n}{\underbrace{\rho_{\theta}\otimes\cdots\otimes\rho_{\theta}}%
}=:\rho_{\theta}^{\otimes n}$, where $\rho_{\theta}$ is drawn from a
\textit{quantum statistical model} $\mathcal{M}:=\{\rho_{\theta}\,;\,\theta
\in\Theta\}$, with $\Theta$'s being an open convex region in $%
\mathbb{R}
^{m}$.

Our purpose is to estimate the true value of $\theta$, based on a measurement
\thinspace$\mathsf{M}^{n}$ acting in $\mathcal{H}^{\otimes n}$. \ Based on the
measurement result $\omega_{n}\in%
\mathbb{R}
^{l_{n}}$, we compute the estimate $T_{n}$ of $\theta$. The pair
$\mathcal{E}_{n}:=\left\{  \,\mathsf{M}^{n},\,T_{n}\right\}  $ \ (or sometimes
the sequence $\left\{  \mathcal{E}_{n}\right\}  _{n=1}^{\infty}$ also) is
called an\ \textit{estimator}. $T_{n}$ is a measurable function of $%
\mathbb{R}
^{l_{n}}$ to $\hat{\Theta}_{n}\subset%
\mathbb{R}
^{m}$. \ The following notations are used: $\mathrm{E}_{\theta}^{\,\mathsf{M}%
^{n}}\left[  f\left(  \omega_{n}\right)  \right]  :=\int f\left(  \omega
_{n}\right)  \mathrm{tr}\,\rho_{\theta}M^{n}\left(  \mathrm{d}\omega
_{n}\right)  $, $\left(  \mathrm{MSE}_{\theta}\left[  \mathcal{E}_{n}\right]
\right)  _{i,j}:=\mathrm{E}_{\theta}^{\mathsf{M}^{n}}\left(  T_{n}^{i}%
-\theta^{i}\right)  \left(  T_{n}^{j}-\theta^{j}\right)  $, $\ \left(
\mathrm{V}_{\theta}\left[  \mathcal{E}_{n}\right]  \right)  _{i,j}%
:=\mathrm{E}_{\theta}^{\mathsf{M}^{n}}\left(  T_{n}^{i}-\mathrm{E}_{\theta
}^{\,\mathsf{M}^{n}}\left[  T_{n}^{i}\right]  \right)  \left(  T_{n}%
^{j}-\mathrm{E}_{\theta}^{\,\mathsf{M}^{n}}\left[  T_{n}^{j}\right]  \right)
$. Below, $\mathrm{Tr}$ denotes the trace over $%
\mathbb{R}
^{m}$, and $\partial_{j}:=\frac{\partial}{\partial\theta^{j}}$. $G_{\theta}$
is a symmetric positive real matrix, and $\theta\rightarrow G_{\theta}$ is
continuously differentiable, $\mathrm{Tr}\,G_{\theta}\leq b_{1}$, and
$\left\vert \mathrm{Tr}\,G_{\theta}-\mathrm{Tr}\,G_{\theta^{\prime}%
}\right\vert \leq b_{1}\left\Vert \theta-\theta^{\prime}\right\Vert $. We also
define $\left(  B_{\theta_{0}}\left[  \mathcal{E}_{n}\right]  \right)
_{j}^{i}:=\left.  \partial_{j}\mathrm{E}_{\theta}^{\,\mathsf{M}^{n}}\left[
T_{n}^{i}\right]  \right\vert _{\theta=\theta_{0}}$. Our interest is the first
order asymptotic term of the weighted mean square error $\varlimsup
_{n\rightarrow\infty}n\mathrm{Tr}\,G_{\theta}\mathrm{MSE}_{\theta}\left[
\mathcal{E}_{n}\right]  $, minimized over \textit{asymptotically unbiased
estimator }, or $\left\{  \mathcal{E}_{n}\right\}  _{n=1}^{\infty}$ with the
following condition:%
\begin{equation}
\lim_{n\rightarrow\infty}\mathrm{E}_{\theta}^{\mathsf{M}^{n}}\left[
T_{n}\right]  =\theta,\quad\lim_{n\rightarrow\infty}\,\left(  B_{\theta
}\left[  \mathcal{E}_{n}\right]  \right)  _{j}^{i}=\delta_{j}^{i}%
,\ \ \forall\theta\in\Theta. \label{asym-unbiased}%
\end{equation}
\vskip-\lastskip In considering (\ref{asym-unbiased}), $\mathrm{E}_{\theta
}^{\,\mathsf{M}^{n}}\left[  T_{n}\right]  $ has to be differentiable, which is
made sure by Lemma\thinspace\ref{lem:est-cont-2}. Use of MSE may be justified
based on the existence of the asymptotic normal efficient estimator, which is
composed in Subsection\thinspace3.3.

Our purpose is to replace this condition by the following tractable condition
without changing the optimal lowerbound to the asymptotic cost: $\mathcal{E}%
_{\theta_{0},n}=\{\mathsf{M}_{\theta_{0}}^{n},\,T_{\theta_{0},n}\}$ is said to
be \textit{locally unbiased at }$\theta_{0}$ if%
\begin{equation}
\mathrm{E}_{\theta_{0}}^{\mathsf{M}_{\theta_{0}}^{n}}\left[  T_{\theta_{0}%
,n}\right]  =\theta_{0},\ \quad\left(  B_{\theta_{0}}\left[  \mathcal{E}%
_{\theta_{0},n}\right]  \right)  _{j}^{i}=\delta_{j}^{i}.
\label{locally-unbiased}%
\end{equation}

\vskip-\lastskip Note that the condition (\ref{locally-unbiased}) is closed at
the point $\theta_{0}$. In the following sections, we prove that minimization
of $\varlimsup_{n\rightarrow\infty}n\mathrm{Tr}\,G_{\theta}\mathrm{MSE}%
_{\theta}\left[  \mathcal{E}_{n}\right]  $ \ over all the asymptotically
unbiased estimators can be reduce to minimization over the locally unbiased
estimators under some proper regularity conditions. \vskip-\lastskip\vskip-5mm

\section*{{\protect\normalsize 3. THE BASIC SETTING }}

\vskip-\lastskip\vskip-3mm

\subsection*{\noindent{\protect\normalsize 3.1. REGULARITY CONDITIONS AND
ASYMPTOTIC CRAMER-RAO BOUND}}

\vskip-3mm Regularity conditions on quantum statistical models and estimators
are listed in Table \ref{table:reg-cond-model}, in which convergence is with
respect to $\left\Vert \cdot\right\Vert _{1}$. $\Diamond_{i,\theta,n}$ is as
defined in Lemma\,\ref{lem:diamond}, and $\Diamond_{i,\theta,n}^{(1)}%
:=\Diamond_{i,\theta,n}\otimes\rho_{\theta}^{\otimes n-1}+\rho_{\theta}%
\otimes\Diamond_{i,\theta,n}\otimes\rho_{\theta}^{\otimes n-2}+\cdots$.

Among the conditions on models, only (M.1) is needed to prove the lowerbound.
\textit{Unless otherwise mentioned, (M.1) are assumed throughout the paper.}
(M.2-3) are necessarily to prove the achievability of the lowerbound. (M.2) is
equivalent to $\left\vert \partial_{i}\mathrm{tr}\,\rho_{\theta}X\right\vert
\leq c\left\vert \mathrm{tr}\,\rho_{\theta}X^{2}\right\vert $ for any bounded Hermitian.

If $\dim\mathcal{H<\infty}$, an example of estimator $\mathcal{\tilde{E}}%
_{n}=\left\{  \mathsf{\tilde{M}}^{n},\tilde{T}_{n}\right\}  $ with (M.3.1-3)
is constructed as follows. Let $l:=\left(  \dim\mathcal{H}\right)  ^{2}$, and
define $\mathbf{e}_{\upsilon}:=\left(  0,\cdots,0,\overset{\upsilon}%
{1},0,\cdots,0\right)  ^{T}\in%
\mathbb{R}
^{l}$. Let $\mathrm{supp}\left(  \tilde{M}\right)  $ be $\left\{
\mathbf{e}_{\upsilon}\right\}  _{\upsilon=1}^{l}$, and let $\left\{  \tilde
{M}\left(  \left\{  \mathbf{e}_{\upsilon}\right\}  \right)  \right\}
_{\upsilon=1}^{l-1}$ be linearly independent. Denoting the $\kappa$-th
measurement result by $\omega_{1,\kappa}$, we can estimate $\mathrm{tr}%
\rho_{\theta}\tilde{M}\left(  \left\{  \mathbf{e}_{\upsilon}\right\}  \right)
$ by the relative frequency of observing $\mathbf{e}_{\upsilon}$, which is
$\upsilon$-th component $\overline{\omega}_{1}^{\upsilon}$ of $\overline
{\omega}_{1}:=\frac{1}{n}\sum_{\kappa=1}^{n}\omega_{1,\kappa}$. Let $\hat
{\rho}$ be a solution to the system of linear equations $\mathrm{tr}\hat{\rho
}\tilde{M}\,\left(  \left\{  \mathbf{e}_{\upsilon}\right\}  \right)
=\overline{\omega}_{1}^{\upsilon}$ ($\upsilon=1$, $\cdots$, $l$), and
$\tilde{T}_{n}$ is defined by $\rho_{\tilde{T}_{n}}=\Pi\left(  \hat{\rho
}\right)  $, where $\Pi$ is a properly defined projection. Also, if $\left\{
\rho_{\theta}\right\}  _{\theta\in\Theta}$ is a smooth submodel of quantum
Gaussian model $\left\{  \sigma_{\eta}\right\}  $, we can compose
$\mathcal{\tilde{E}}_{n}$ based on the estimator $\hat{\eta}_{n}$ of $\eta$ by
$\rho_{\tilde{T}_{n}}=\Pi\left(  \sigma_{\hat{\eta}_{n}}\right)  $, with
proerly defined projection $\Pi$.

Both of them has the following property. $\left\{  \rho_{\theta}\right\}
_{\theta\in\Theta}$ is a somooth submaniforld of a larger quantum state model
$\left\{  \sigma_{\eta}\right\}  $, where $\eta$ has consistent estimator in
the form of $\hat{\eta}_{n}=\frac{1}{n}\sum_{\kappa=1}^{n}\omega_{1,\kappa}$,
where $\omega_{1,\kappa}$ is the data obtained by application of
$\mathsf{\tilde{M}}$ on the $\kappa$-th sample. Suppose that $\eta=\left(
\theta,\zeta\right)  $, and $\rho_{\theta}=\sigma_{\theta,\zeta\left(
\theta\right)  }$. Moreover, we suppose that $\zeta\left(  \theta\right)  $ is
uniformly continuous in $\theta$. Then, $\tilde{T}_{n}:=\left(  \hat{\eta}%
_{n}^{1},\cdots,\hat{\eta}_{n}^{m}\right)  $ satisfies the requirements.

As for the estimators, besides (\ref{asym-unbiased}), we suppose
$\mathcal{E}_{n}\mathcal{=}\left\{  \mathsf{M}^{n},T_{n}\right\}  $ satisfies
(E) in Table\thinspace\thinspace\ref{table:reg-cond-op-model}\ for all $n$.
(E') is used to characterize lowerbound to the asymptotic cost. Observe that
(E')$\Longrightarrow$(E).

We define the \textit{\ asymptotic quantum Cramer-Rao type bound }$C_{\theta
}^{Q}\left(  G_{\theta},\mathcal{M}\right)  $ as\newline$\underset
{n\rightarrow\infty}{\varlimsup}\inf\left\{  n\mathrm{Tr}\,G_{\theta
}\mathrm{MSE}_{\theta}\left[  \mathcal{E}_{n}\right]  \,;\mathsf{M}^{n}\text{
in }\mathcal{H}^{\otimes n}\text{, (\ref{asym-unbiased}), (E)}\right\}  $. In
the succeeding subsections, the following theorem will be proved. In the
remaining of this subsection, some technical lemmas will be shown.%

\begin{table}[tbp] \centering
\begin{tabular}
[c]{|l|}\hline
(M.1)\quad$\partial_{i}\rho_{\theta}$ and $\partial_{i}\partial_{j}%
\rho_{\theta}$ \ exist and are locally uniformly continuous. $\left\Vert
\partial_{i}\rho_{\theta_{0}}\right\Vert $,$\left\Vert \partial_{i}%
\partial_{j}\rho_{\theta}\right\Vert _{1}$ $\leq a_{1}<\infty$.\\
(M.2)\quad$\exists L_{\theta,i}$: Hermitian and $\partial_{i}\rho_{\theta
}=\frac{1}{2}\left(  L_{\theta,i}\rho_{\theta}+\rho_{\theta}L_{\theta
,i}\right)  $, and $\mathrm{tr}\,\rho_{\theta}\left(  L_{\theta,i}\right)
^{2}<\infty$, $\forall\theta\in\Theta$.\\
(M.3)\quad There is an estimator $\mathcal{\tilde{E}}_{n}=\left\{
\mathsf{\tilde{M}}^{n},\tilde{T}_{n}\right\}  $ in $\mathcal{H}^{\otimes n}$ ,
such that\\
\ \ (M.3.1)\quad(\ref{asym-unbiased}) and (E) are satisfied. \ \\
\ \ (M.3.2)\quad$\mathrm{E}_{\theta}^{\mathsf{\tilde{M}}^{n}}\left\Vert
\tilde{T}_{n}-\theta\right\Vert ^{4}\leq\frac{D_{\theta,2}}{n^{2}}$,
$\forall\theta\in\Theta$, $\exists$ $D_{\theta,2}$.\\
\ \ (M.3.3)\quad$\mathsf{\tilde{M}}^{n}$ is $n$ times repetition of a
measurement $\mathsf{\tilde{M}}$ in $\mathcal{H}$, producing the data
$x_{\kappa}\in%
\mathbb{R}
^{l}$.\\\hline\hline
(E)$\quad\exists a_{4,n}$, $\forall\theta\in\Theta$, $\int\left\Vert
T_{n}\left(  \omega_{n}\right)  -\theta\right\Vert \mathrm{tr}\,\Diamond
_{\theta,n}^{\left(  1\right)  }M^{n}\left(  \mathrm{d}\omega_{n}\right)  \leq
na_{1}a_{4,n}$,\quad$\int\left\Vert T_{n}\left(  \omega_{n}\right)
-\theta\right\Vert ^{2}\mathrm{tr}\rho_{\theta}^{\otimes n}M^{n}\left(
\mathrm{d}\omega_{n}\right)  \leq na_{4,n}^{2}$.\\\hline
(E')\quad$T_{n}$ takes values in $\hat{\Theta}_{T_{n}}$, with $\sup
_{\theta,\theta^{\prime}\in\hat{\Theta}_{T_{n}}}\left\Vert \theta
-\theta^{\prime}\right\Vert \leq a_{4,n}<\infty$.\\\hline
\end{tabular}
\caption{Regularity conditions on quantum statistical models (M.1-4) and
estimators (E), (E')}\label{table:reg-cond-model}%
\end{table}%

\begin{theorem}
\label{th:cr}Suppose (M.1-3) hold. Then,
\begin{align}
&  C_{\theta}^{Q}\left(  G_{\theta},\mathcal{M}\right)  =\lim_{n\rightarrow
\infty}\inf\left\{  n\mathrm{Tr}\,G_{\theta}\mathrm{V}_{\theta}\left[
\mathcal{E}_{\theta,n}\right]  \,\text{\thinspace};\text{ }\mathsf{M}%
^{n}\text{ in }\mathcal{H}^{\otimes n}\text{, (\ref{locally-unbiased}), (E')
}\right\}  ,\label{cr-q1}\\
&  =\lim_{n\rightarrow\infty}\inf\left\{  n\mathrm{Tr}\,G_{\theta}%
\mathrm{V}_{\theta}\left[  \mathcal{E}_{\theta,n}\right]  \,\text{\thinspace
};\text{ }\mathsf{M}^{n}\text{ in }\mathcal{H}^{\otimes n}\text{,
(\ref{locally-unbiased}), (E) }\right\}  . \label{cr-q2}%
\end{align}

\end{theorem}

\begin{lemma}
\label{lem:est-cont-2}(E) and (M.1) imply the existence of $\partial
_{j}\mathrm{E}_{\theta}^{\mathsf{M}^{n}}\left[  T_{n}^{i}\right]  $
and$\ \ \partial_{j}\mathrm{E}_{\theta}^{\mathsf{M}^{n}}\left[  T_{n}%
^{i}\right]  =\int T_{n}^{i}\left(  \omega_{n}\right)  \mathrm{tr}%
\,\partial_{j}\rho_{\theta}M^{n}\left(  \mathrm{d}\omega_{n}\right)  $.
\end{lemma}

\begin{proof}
Due to Lemma\thinspace\ref{lem:diamond}, this Lemma is equivalent to
Proposition VI.2.2 of Holevo\thinspace(1982).
\end{proof}

\begin{lemma}
\label{lem:diamond}(M.1) implies that $\exists a_{1}\exists a_{2}\forall i$,
$\forall\theta$, $\theta_{0}\in\Theta\,$and $\left\vert \theta^{i}-\theta
_{0}^{i}\right\vert <a_{2}$, $\theta^{j}=\theta_{0}^{j}$ ($j\neq i$),
$\exists\Diamond_{i,\theta}$ such that $\,\left\vert \partial_{i}\rho
_{\theta_{0}}\right\vert \,\leq\Diamond_{i,\theta}\,\,$, $\mathrm{tr}%
\,\Diamond_{i,\theta}\,\leq a_{1}<\infty$.
\end{lemma}

\begin{proof}
Since $\partial_{i}\rho_{\theta_{0}}=\partial_{i}\rho_{\theta-a_{2}%
\mathbf{e}_{i}}+\int_{x=\theta-a_{2}\mathbf{e}_{i}}^{\theta_{0}}\partial
_{i}^{2}\rho_{x}\mathrm{d}x$, $\Diamond_{i,\theta}:=\left\vert \partial
_{i}\rho_{\theta-a_{2}\mathbf{e}_{i}}\right\vert +\int_{x=\theta
-a_{2}\mathbf{e}_{i}}^{\theta+a_{2}\mathbf{e}_{i}}\left\vert \partial_{i}%
^{2}\rho_{x}\right\vert \mathrm{d}x$, if exists in the sense of Bochner,
satisfies requirement. This is true since $\left\Vert \partial_{i}^{2}%
\rho_{\theta}\right\Vert _{1}$ is continuous in $\theta$ (hence, measurable
and integrable over the finite interval).
\end{proof}

\begin{lemma}
\label{lem:est-cont}(E'), combined with (M.1), implies
\begin{align}
&  \partial_{j}^{t_{j}}\partial_{k}^{t_{k}}\mathrm{E}_{\theta}^{\mathsf{M}%
^{n}}\left[  T_{n}\right]  =\int T_{n}\left(  \omega_{n}\right)
\mathrm{tr}\,\partial_{j}^{t_{j}}\partial_{k}^{t_{k}}\rho_{\theta}^{\otimes
n}M^{n}\left(  \mathrm{d}\omega_{n}\right)  \quad(t_{j},t_{k}\in\left\{
0,1\right\}  ),\label{dE=Ed}\\
&  \left\vert \mathrm{Tr}\,G_{\theta}\mathrm{V}_{\theta}\left[  \mathcal{E}%
_{n}\right]  -\mathrm{Tr}\,G_{\theta^{\prime}}\mathrm{V}_{\theta^{\prime}%
}\left[  \mathcal{E}_{n}\right]  \right\vert \leq\left(  na_{1}+1\right)
b_{1}\left(  a_{4,n}\right)  ^{2}\left\Vert \theta-\theta^{\prime}\right\Vert
,\label{GV-GV'}\\
&  \left\Vert \mathrm{E}_{\theta}^{\mathsf{M}^{n}}\left[  T_{n}\right]
-\mathrm{E}_{\theta^{\prime}}^{\mathsf{M}^{n}}\left[  T_{n}\right]
\right\Vert \leq m^{2}na_{4,n}a_{1}\left\Vert \theta-\theta^{\prime
}\right\Vert ,\label{E-E'}\\
&  \left\Vert \partial_{j}\mathrm{E}_{\theta}^{\mathsf{M}^{n}}\left[
T_{n}\right]  -\partial_{j}\mathrm{E}_{\theta^{\prime}}^{\mathsf{M}^{n}%
}\left[  T_{n}\right]  \right\Vert \leq m^{2}n^{2}a_{4,n}a_{1}^{2}\left\Vert
\theta-\theta^{\prime}\right\Vert ,\label{dE-dE'-2}\\
&  \lim_{\theta\rightarrow\theta_{0}}\left(  B_{\theta_{0}}\left[
\mathcal{E}_{\theta,n}\right]  \right)  _{j}^{i}=\delta_{j}^{i},\text{ where
}\left\{  \mathcal{E}_{\theta_{0},n}\right\}  _{\theta_{0}\in\Theta}\text{
satisfies (\ref{locally-unbiased})}.\text{ } \label{Bxy}%
\end{align}

\end{lemma}

\begin{proof}
(E') implies $\left\vert \int T_{n}^{i}\left(  \omega_{n}\right)
\mathrm{tr}\,\tau M^{n}\left(  \mathrm{d}\omega_{n}\right)  \right\vert
\leq\left\vert \int\left\vert T_{n}^{i}\left(  \omega_{n}\right)  \right\vert
\mathrm{tr}\,\tau M^{n}\left(  \mathrm{d}\omega_{n}\right)  \right\vert
\leq\left\Vert \tau\right\Vert _{1}a_{4,n}$. Therefore, the map $\tau
\rightarrow\int T_{n}^{i}\left(  \omega_{n}\right)  \mathrm{tr}\,\tau
M^{n}\left(  \mathrm{d}\omega_{n}\right)  $ is a continuous linear functional,
and is exchangeable with $\lim$. Therefore, the first two identities follow.
To show (\ref{E-E'}), apply the mean value theorem to the function
$\theta\rightarrow\mathrm{E}_{\theta}^{\mathsf{M}^{n}}\left[  T_{n}\right]  $.
Due to (\ref{dE=Ed}), we obtain $\left\vert \mathrm{E}_{\theta}^{\mathsf{M}%
^{n}}\left[  T_{n}^{i}\right]  -\mathrm{E}_{\theta^{\prime}}^{\mathsf{M}^{n}%
}\left[  T_{n}^{i}\right]  \right\vert \leq\sum_{j=1}^{m}\left\vert \int
T_{n}^{i}\left(  \omega_{n}\right)  \mathrm{tr}\,\partial_{j}\rho
_{\theta_{\ast}}^{\otimes n}M^{n}\left(  \mathrm{d}\omega_{n}\right)
\right\vert \left\vert \theta^{j}-\theta^{\prime j}\right\vert $. Therefore,
due to (M.1) and Lemma\thinspace\ref{lem:diamond}, we have (\ref{E-E'}).
(\ref{dE-dE'-2}) is shown similarly. To show (\ref{Bxy}), observe
\begin{align*}
\left\vert \left(  B_{\theta_{0}}\left[  \mathcal{E}_{\theta,n}\right]
\right)  _{j}^{i}-\delta_{j}^{i}\right\vert  &  =\left\vert \left(
B_{\theta_{0}}\left[  \mathcal{E}_{\theta,n}\right]  \right)  _{j}^{i}-\left(
B_{\theta_{0}}\left[  \mathcal{E}_{\theta_{0},n}\right]  \right)  _{j}%
^{i}\right\vert \\
&  \leq\left\vert \left(  B_{\theta_{0}}\left[  \mathcal{E}_{\theta,n}\right]
\right)  _{j}^{i}-\left(  B_{\theta}\left[  \mathcal{E}_{\theta,n}\right]
\right)  _{j}^{i}\right\vert +\left\vert \left(  B_{\theta}\left[
\mathcal{E}_{\theta,n}\right]  \right)  _{j}^{i}-\left(  B_{\theta_{0}}\left[
\mathcal{E}_{\theta_{0},n}\right]  \right)  _{j}^{i}\right\vert =\left\vert
\left(  B_{\theta_{0}}\left[  \mathcal{E}_{\theta,n}\right]  \right)  _{j}%
^{i}-\left(  B_{\theta}\left[  \mathcal{E}_{\theta,n}\right]  \right)
_{j}^{i}\right\vert .
\end{align*}%
\vskip-\lastskip
Due to (\ref{dE-dE'-2}), we have (\ref{Bxy}).
\end{proof}

\vskip-10mm

\subsection*{\noindent{\protect\normalsize 3.2 LOWERBOUND AND (\ref{cr-q1}%
)=(\ref{cr-q2})}}

\vskip-3mm First we prove that the RHS of {\normalsize (\ref{cr-q2}) is a
lowerbound to }$C_{\theta}^{Q}\left(  G_{\theta},\mathcal{M}\right)  $. Define
locally unbiased estimator $\mathcal{E}_{\theta,n}=\{$\thinspace
$\mathsf{M}^{n},T_{\theta,n}\,\}$ by $T_{n}=B_{\theta}\left[  \mathcal{E}%
_{n}\right]  \left(  T_{\theta,n}-\theta\right)  +\mathrm{E}_{\theta
}^{\,\mathsf{M}^{n}}\left[  T_{n}\right]  $. Obviously,

$\quad\quad\quad\quad\quad\quad\quad n\mathrm{Tr}\,G_{\theta}\mathrm{MSE}%
_{\theta}\left[  \mathcal{E}_{n}\right]  \geq n\mathrm{Tr}\,G_{\theta
}\mathrm{V}_{\theta}\left[  \mathcal{E}_{n}\right]  =n\mathrm{Tr}\,G_{\theta
}B_{\theta}\left[  \mathcal{E}_{n}\right]  \mathrm{V}_{\theta}\left[
\mathcal{E}_{\theta,n}\right]  B_{\theta}\left[  \mathcal{E}_{n}\right]  ^{T}%
$, \newline and letting $n\rightarrow\infty$, we have our assertion due to
(\ref{asym-unbiased}).

Below, we prove {\normalsize (\ref{cr-q1})=(\ref{cr-q2}). Since (E') implies
(E), it suffices to show (\ref{cr-q1})}${\normalsize \leq}$%
{\normalsize (\ref{cr-q2}). } Suppose $\mathcal{E}_{\theta,n}$ satisfies (E)
and (\ref{locally-unbiased}). Let $S_{\theta,n}^{L}$ $:=T_{\theta,n}$ in
$\left\Vert T_{\theta,n}-\theta\right\Vert \leq L$-case and $S_{\theta,n}%
^{L}:=\theta$ otherwise. Let $\mathcal{F}_{\theta,n}^{L}:=\left\{
\mathsf{M}_{\theta}^{n},T_{\theta,n}^{L}\right\}  $, and let $\ \mathcal{E}%
_{\theta,n}^{L}=\left\{  \mathsf{M}_{\theta}^{n},T_{\theta,n}^{L}\right\}  $
be a locally unbiased estimator with $T_{\theta,n}^{L}=B_{\theta}\left[
\mathcal{F}_{\theta,n}^{L}\right]  ^{-1}\left(  S_{\theta,n}^{L}%
-\mathrm{E}_{\theta_{0}}^{\,\mathsf{M}^{n}}\left[  S_{\theta,n}^{L}\right]
\right)  +\theta$. Obviously, $\mathcal{E}_{\theta,n}^{L}$ satisfies (E').
Also, due to Lemma\thinspace\ref{lem:est-cont-2}, Lemma\thinspace
\ref{lem:diamond}, and (E), we have
\begin{align*}
\left\vert \left(  B_{\theta}\left[  \mathcal{\tilde{F}}_{\theta,n}%
^{L}\right]  \right)  _{j}^{i}-\delta_{j}^{i}\right\vert  &  =\left\vert
\partial_{i}\int_{\left\Vert T_{\theta,n}-\theta\right\Vert >L}\left(
T_{\theta,n}^{j}-\theta^{j}\right)  P_{\theta}^{\mathsf{M}_{\theta}^{n}%
}\left(  \mathrm{d}\omega_{n}\right)  \right\vert =\left\vert \int_{\left\Vert
T_{\theta,n}-\theta\right\Vert >L}\left(  T_{\theta,n}^{j}-\theta^{j}\right)
\mathrm{\mathrm{tr}\,}\partial_{i}\rho_{\theta}^{\otimes n}M_{\theta}%
^{n}\left(  \mathrm{d}\omega_{n}\right)  \right\vert \\
&  \leq\int_{\left\Vert T_{\theta,n}-\theta\right\Vert >L}\left\Vert
T_{\theta,n}-\theta\right\Vert \mathrm{\mathrm{tr}\,}\Diamond_{i,\theta
,n}^{\left(  1\right)  }M_{\theta}^{n}\left(  \mathrm{d}\omega_{n}\right)
\rightarrow0\quad\left(  L\rightarrow\infty\right)  .
\end{align*}%
\vskip-\lastskip
Therefore, $\forall\varepsilon>0\exists L$,

$\mathrm{Tr}\,G_{\theta}\mathrm{V}_{\theta}\left[  \mathcal{E}_{\theta
,n}\right]  \geq\mathrm{Tr}\,G_{\theta}\mathrm{V}_{\theta}\left[
\mathcal{F}_{\theta,n}^{L}\right]  =\mathrm{Tr}\,G_{\theta}B_{\theta}\left[
\mathcal{F}_{\theta,n}^{L}\right]  \mathrm{V}_{\theta}\left[  \mathcal{E}%
_{\theta,n}^{L}\right]  B_{\theta}\left[  \mathcal{F}_{\theta,n}^{L}\right]
^{T}\geq\mathrm{Tr}\,G_{\theta}\mathrm{V}_{\theta}\left[  \mathcal{E}%
_{\theta,n}^{L}\right]  -\varepsilon.$

\noindent\ Taking infimum of the both ends, we have {\normalsize (\ref{cr-q1}%
)}${\normalsize \leq}${\normalsize (\ref{cr-q2}). }\vskip-10mm

\subsection*{{\protect\normalsize \noindent}\noindent{\protect\normalsize 3.3
ACHIEVABILITY}}

\vskip-3mm Based on $\{\mathcal{E}_{\theta,n_{1}}\}_{\theta\in\Theta}$
$=\{\mathsf{M}_{\theta}^{n_{1}},T_{\theta,n_{1}}\,\}_{\theta\in\Theta}$ \ such
that (\ref{locally-unbiased}) and (E') with $n=n_{1}$ are satisfied, we
construct a good estimator $\mathcal{E}_{n}^{n_{1}}$ with 2 steps in the
following. Given $\rho_{\theta}^{\otimes n}$, invest $\rho_{\theta}^{\otimes
n_{0}}$ to obtain the data $\vec{\omega}_{1}:=(\omega_{1,1},\cdots
,\omega_{1,n_{0}})$, where $\omega_{1,i}\in%
\mathbb{R}
^{l}$. Based on the data, we compute the estimator $\theta_{0}=\widetilde
{T}_{n_{0}}\left(  \vec{\omega}_{1}\right)  $. \ Now, we divide $\rho_{\theta
}^{\otimes n-n_{0}}$ into the ensembles each with $n_{1}$ copies. The number
of ensemble, $\frac{n-n_{0}}{n_{1}}$, is denoted by $n_{2}$. Here, $n_{0}$ and
$n_{2}$ are chosen so that $n_{0}=n_{2}^{3/4}$ is satisfied. We apply
$\mathsf{M}_{\theta_{0}}^{n_{1}}$ \ to $\ $each ensemble $\rho_{\theta
}^{\otimes n_{1}}$, obtain the data $\omega_{2,1}$, $\cdots$, $\omega
_{2,n_{2}}$($\in%
\mathbb{R}
^{l_{n_{1}}}$) and compute
\begin{equation}
T_{n}^{n_{1}}:=\frac{1}{n_{2}}\sum_{\kappa=1}^{n_{2}}T_{\theta_{0},n_{1}%
}\,\left(  \omega_{2,\kappa}\right)  . \label{def-est}%
\end{equation}%
\vskip-\lastskip
The measurement defined above is denoted by $\mathsf{M}^{n_{1},n}$.

\begin{lemma}
\label{lem:est-compose}Suppose that (M.1,3) hold. Suppose also that the family
$\{\mathcal{E}_{\theta,n_{1}}\}_{\theta\in\Theta}$ satisfies
(\ref{locally-unbiased}) and (E') with $n=n_{1}$, $\forall\theta\in\Theta$.
Then $\mathcal{E}_{n}^{n_{1}}$ constructed above satisfies $\underset
{n_{2}\rightarrow\infty}{\lim}n\mathrm{Tr}G\,_{\theta}\mathrm{MSE}_{\theta
}\left[  \mathcal{E}_{n}^{n_{1}}\right]  \leq n_{1}\underset{\theta
_{0}\rightarrow\theta}{\varlimsup}\mathrm{Tr}G\,_{\theta}\mathrm{V}_{\theta
}\left[  \mathcal{E}_{\theta_{0},n_{1}}\right]  $.
\end{lemma}

\begin{proof}
Applying mean value theorem to the function $\theta\rightarrow\mathrm{E}%
_{\theta}^{\mathsf{M}_{\theta_{0}}^{n_{1}}}\left[  T_{\theta_{0},n_{1}}%
^{i}\right]  $, we have
\begin{equation}
\mathrm{E}_{\theta}^{\mathsf{M}_{\theta_{0}}^{n_{1}}}\left[  T_{\theta
_{0},n_{1}}^{i}\right]  =\mathrm{E}_{\theta_{0}}^{\mathsf{M}_{\theta_{0}%
}^{n_{1}}}\left[  T_{\theta_{0},n_{1}}^{i}\right]  +\sum_{j=1}^{m}\left(
\theta^{j}-\theta_{0}^{j}\right)  \left.  \partial_{j}\mathrm{E}_{\theta
}^{\mathsf{M}_{\theta_{0}}}\left[  T_{\theta_{0},n_{1}}^{i}\right]
\right\vert _{\theta=\theta_{0}}+\gamma_{\theta,\theta_{0}}^{n_{1},i}%
=\theta_{0}^{i}+\left(  \theta^{i}-\theta_{0}^{i}\right)  +\gamma
_{\theta,\theta_{0}}^{n_{1},i}=\theta^{i}+\gamma_{\theta,\theta_{0}}^{n_{1},i}
\label{tildeT-bias}%
\end{equation}
\vskip-\lastskip where $\gamma_{\theta,\theta_{0}}^{n_{1},i}$ is the reminder
term. With the help of (\ref{dE=Ed}) and (M.1),%
\vskip-\lastskip
\begin{equation}
\left\vert \gamma_{\theta,\theta_{0}}^{n_{1},i}\right\vert =\frac{1}%
{2}\left\vert \sum_{j,k=1}^{m}\left(  \theta^{j}-\theta_{0}^{j}\right)
\left(  \theta^{k}-\theta_{0}^{k}\right)  \int T_{\theta_{0},n_{1}}^{i}\left(
\omega\right)  \mathrm{tr}\,\partial_{j}\partial_{k}\rho_{\theta^{\prime}%
}^{\otimes n_{1}}M_{\theta_{0}}^{n_{1}}\left(  \mathrm{d}\omega\right)
\right\vert \leq n_{1}^{2}m^{2}a_{1}^{2}a_{4,n_{1}}\left\Vert \theta
_{0}-\theta\right\Vert ^{2}, \label{|gamma|<}%
\end{equation}%
\vskip-\lastskip
where $\theta^{\prime}$ lies between $\theta_{0}$ and $\theta$. Since MSE is
the sum of the variance and square of the bias, we have
\begin{align*}
&  \mathrm{Tr}G_{\theta}\mathrm{MSE}_{\theta}\left[  \left.  \mathcal{E}%
_{n}^{n_{1}}\right\vert \widetilde{T}_{n_{0}}=\theta_{0}\right]
=\mathrm{Tr}G_{\theta}\mathrm{V}_{\theta}\left[  \left.  \mathcal{E}%
_{n}^{n_{1}}\right\vert \widetilde{T}_{n_{0}}=\theta_{0}\right]  +\sum
_{i,j=1}^{m}\left(  G_{\theta}\right)  _{i,j}\gamma_{\theta,\theta_{0}}%
^{n_{1},i}\,\gamma_{\theta,\theta_{0}}^{n_{1},j}\\
&  \leq\frac{1}{n_{2}}\mathrm{Tr}G_{\theta}\mathrm{V}_{\theta}\left[
\mathcal{E}_{\theta_{0},n_{1}}\right]  +n_{1}^{4}m^{4}\left(  a_{1}%
^{2}a_{4,n_{1}}\right)  ^{2}\mathrm{Tr}G_{\theta}\left\Vert \theta_{0}%
-\theta\right\Vert ^{4}.
\end{align*}%
\vskip-\lastskip
Taking average over $\widetilde{T}_{n_{0}}$ of the left most and the right
most end,
\begin{align*}
&  \lim_{n_{2}\rightarrow\infty}n\mathrm{Tr}G_{\theta}\mathrm{MSE}_{\theta
}\left[  \mathcal{E}_{n}^{n_{1}}\right]  \leq\lim_{n_{2}\rightarrow\infty
}\left[  \frac{n}{n_{2}}\mathrm{E}_{\theta}^{\mathsf{\tilde{M}}^{n_{0}}%
}\mathrm{Tr}G_{\theta}\mathrm{V}_{\theta}\left[  \mathcal{E}_{\tilde{T}%
_{n_{0}},n_{1}}\right]  +nn_{1}^{4}m^{4}\left(  a_{1}^{2}a_{4,n_{1}}\right)
^{2}\mathrm{Tr}G_{\theta}\mathrm{E}_{\theta}^{\mathsf{\tilde{M}}^{n_{0}}%
}\left\Vert \tilde{T}_{n_{0}}-\theta\right\Vert ^{4}\right] \\
&  \underset{\text{(i)}}{\leq}\lim_{n_{2}\rightarrow\infty}\sup_{\theta
_{0}:\left\Vert \theta_{0}-\theta\right\Vert <\varepsilon}n_{1}\mathrm{Tr}%
G_{\theta}\mathrm{V}_{\theta}\left[  \mathcal{E}_{\theta_{0},n_{1}}\right]
+\frac{D_{\theta,2}}{\varepsilon^{4}n_{0}^{2}}\sup_{\theta_{0}\in%
\mathbb{R}
^{m}}\mathrm{Tr}G_{\theta}\mathrm{V}_{\theta}\left[  \mathcal{E}_{\theta
_{0},n_{1}}\right]  +nn_{1}^{4}m^{4}\left(  a_{1}^{2}a_{4,n_{1}}\right)
^{2}\mathrm{Tr}G_{\theta}\mathrm{E}_{\theta}^{\mathsf{\tilde{M}}^{n_{0}}%
}\left\Vert \tilde{T}_{n_{0}}-\theta\right\Vert ^{4}\\
&  \underset{\text{(ii)}}{\leq}\lim_{n_{2}\rightarrow\infty}\sup_{\theta
_{0}:\left\Vert \theta_{0}-\theta\right\Vert <\varepsilon}n_{1}\mathrm{Tr}%
G_{\theta}\mathrm{V}_{\theta}\left[  \mathcal{E}_{\theta_{0},n_{1}}\right]
+\frac{D_{\theta,2}}{\varepsilon^{4}n_{2}^{3/2}}\left(  a_{4,n_{1}}\right)
^{2}\mathrm{Tr}G_{\theta}+\lim_{n_{2}\rightarrow\infty}\left(  n_{2}%
n_{1}+n_{0}\right)  n_{1}^{2}m^{4}\left(  a_{1}a_{4,n_{1}}\right)  ^{2}%
\frac{D_{\theta,2}}{n_{2}^{3/2}}\mathrm{Tr}G_{\theta}\\
&  =\sup_{\theta_{0}:\left\Vert \theta_{0}-\theta\right\Vert <\varepsilon
}n_{1}\mathrm{Tr}G_{\theta}\mathrm{V}_{\theta}\left[  \mathcal{E}_{\theta
_{0},n_{1}}\right]  .
\end{align*}
Here (i) is due to $P_{\theta}^{\mathsf{\tilde{M}}^{n_{0}}}\left\{  \left\Vert
\widetilde{T}_{n_{0}}-\theta\right\Vert \geq\varepsilon\right\}  \leq
\frac{D_{\theta,2}}{\varepsilon^{4}n_{0}^{2}}$ which follows from (M.3.2) and
Chebyshev's inequality, and (ii) is due to (M.3.2). Since $\varepsilon>0$ is
arbitrary, the lemma holds.
\end{proof}

\begin{lemma}
\label{lem:est-asym-unbiased}Suppose that (M.1-3) hold. Then $\left\{
\mathcal{E}_{n}^{n_{1}}\right\}  _{n=1}^{\infty}$ satisfies (E) and
(\ref{asym-unbiased}).
\end{lemma}

\begin{proof}
Observe $\left\Vert T_{n}^{n_{1}}-\theta\right\Vert \leq\left\Vert \tilde
{T}_{n_{0}}-\theta\right\Vert +a_{4,n}$ holds. Since $\tilde{T}_{n_{0}}$
satisfies (E) due to (M.3.1), $\left\{  \mathcal{E}_{n}^{n_{1}}\right\}
_{n=1}^{\infty}$ satisfies (E), also.

Observe\thinspace\
\begin{align*}
\left\vert \mathrm{E}_{\theta_{0}}^{\mathsf{M}^{n_{1},n}}\left[  T_{n}%
^{n_{1},j}-\theta_{0}^{j}\right]  \right\vert  &  \leq\mathrm{E}_{\theta_{0}%
}^{\mathsf{\tilde{M}}^{n_{0}}}\left\vert \mathrm{E}_{\theta_{0}}%
^{\mathsf{M}_{\tilde{T}_{n_{0}}}^{n_{1}}}\left[  T_{\tilde{T}_{n_{0}},n_{1}%
}^{,j}-\theta_{0}^{j}\right]  \right\vert \underset{\text{(i)}}{\leq
}\mathrm{E}_{\theta_{0}}^{\mathsf{\tilde{M}}^{n_{0}}}\left\vert \gamma
_{\theta_{0},\tilde{T}_{n_{0}}}^{n_{1},j}\right\vert \\
&  \underset{\text{(ii)}}{\leq}n_{1}^{2}m^{2}a_{1}a_{4,n_{1}}\mathrm{E}%
_{\theta_{0}}^{\mathsf{\tilde{M}}^{n_{0}}}\left\Vert \tilde{T}_{n_{0}}%
-\theta_{0}\right\Vert ^{2}\underset{\text{(iii)}}{\leq}n_{1}^{2}m^{2}%
a_{1}a_{4,n_{1}}\sqrt{\mathrm{E}_{\theta_{0}}^{\mathsf{\tilde{M}}^{n_{0}}%
}\left\Vert \tilde{T}_{n_{0}}-\theta_{0}\right\Vert ^{4}}\underset
{\text{(iv)}}{\rightarrow}0.
\end{align*}%
\vskip-\lastskip
Here, (i), (ii), (iii), and (iv) is due to \ (\ref{tildeT-bias}),
(\ref{|gamma|<}), concavity of $\sqrt{x}$, and (M.3.2), respectively.
Therefore, $\mathrm{E}_{\theta_{0}}^{\mathsf{M}^{n_{1},n}}\left[  T_{n}%
^{n_{1},j}\right]  \rightarrow\theta_{0}^{j}$. $\left(  B_{\theta_{0}}\left[
\mathcal{E}_{n}\right]  \right)  _{j}^{i}\rightarrow\delta_{j}^{i}$ is proved
as follows. In Subsection 4.3 right after the statement of Lemma\thinspace
\ref{lem:leibniz}, we will prove
\begin{equation}
\partial_{i}\left(  \mathrm{E}_{\theta}^{\mathsf{\tilde{M}}^{n_{0}}}%
\mathrm{E}_{\theta}^{\mathsf{M}_{\tilde{T}_{n_{0}}}^{n_{1}}}\left[
T_{\tilde{T}_{n_{0}},n_{1}}^{,j}\right]  \right)  _{\theta=\theta_{0}%
}=\partial_{i}\left(  \mathrm{E}_{\theta}^{\mathsf{\tilde{M}}^{n_{0}}%
}\mathrm{E}_{\theta_{0}}^{\mathsf{M}_{\tilde{T}_{n_{0}}}^{n_{1}}}\left[
T_{\tilde{T}_{n_{0}},n_{1}}^{j}\right]  \right)  _{\theta=\theta_{0}%
}+\mathrm{E}_{\theta_{0}}^{\mathsf{\tilde{M}}^{n_{0}}}\partial_{i}%
\mathrm{E}_{\theta}^{\mathsf{M}_{\tilde{T}_{n_{0}}}^{n_{1}}}\left[
T_{\tilde{T}_{n_{0}},n_{1}}^{j}\right]  _{\theta=\theta_{0}}.
\label{leibniz-pre}%
\end{equation}%
\vskip-\lastskip
Defining $L_{\theta,i}^{n}:=L_{\theta,i}\otimes\mathbf{1}^{\otimes
n-1}+\mathbf{1}\otimes L_{\theta,i}\otimes\mathbf{1}^{\otimes n-2}%
+\cdots+\mathbf{1}^{\otimes n-1}\otimes L_{\theta,i}$ , we have $\partial
_{i}\rho_{\theta}^{\otimes n}=\frac{1}{2}\left(  L_{\theta,i}^{n}\rho_{\theta
}^{\otimes n}+\rho_{\theta}^{\otimes n}L_{\theta,i}^{n}\right)  $,
$\mathrm{\mathrm{tr\,}}\rho_{\theta}^{\otimes n}\left(  L_{\theta,i}%
^{n}\right)  ^{2}=n\mathrm{\mathrm{tr\,}}\rho_{\theta}\left(  L_{\theta
,i}\right)  ^{2}$, and%
\begin{equation}
\partial_{i}\mathrm{tr}\,\rho_{\theta}^{\otimes n}A=\mathrm{tr}\,\partial
_{i}\rho_{\theta}^{\otimes n}A\,=\Re\mathrm{tr}\,\rho_{\theta}^{\otimes
n}AL_{\theta,i}^{n},\text{ }\,\,\forall A\text{: bounded Hermitian},
\label{dE}%
\end{equation}
\vskip-\lastskip where the first identity is due to the continuity of linear
functional $X\rightarrow\mathrm{tr}\,XA$ (e.g., Theorem\thinspace II.7.2 of
Holevo\thinspace(1982)\thinspace). (\ref{dE}), in combination with Schwartz's
inequality, leads to $\left\vert \partial_{i}\mathrm{tr}\,\rho_{\theta
}^{\otimes n}X\right\vert \leq n\mathrm{\mathrm{tr\,}}\rho_{\theta}\left(
L_{\theta,i}\right)  ^{2}\mathrm{tr}\,\rho_{\theta}^{\otimes n}\mathrm{tr}%
\,X^{2}$. Observe $\left\vert \gamma_{\theta_{0},\tilde{T}_{n_{0}}}^{n_{1}%
,j}\right\vert =\left\vert \mathrm{E}_{\theta_{0}}^{\mathsf{M}_{\tilde
{T}_{n_{0}}}^{n_{1}}}\left[  T_{\tilde{T}_{n_{0}},n_{1}}^{,j}\right]
-\theta_{0}^{j}\right\vert \leq\left\vert \tilde{T}_{n_{0}}-\theta_{0}%
^{j}\right\vert +a_{4,n_{1}}$. Hence. due to (\ref{|gamma|<}), (M.3.1), and
Lemma\thinspace\ref{lem:est-cont-2}, we have $\partial_{i}\left[
\mathrm{E}_{\theta}^{\mathsf{\tilde{M}}^{n_{0}}}\gamma_{\theta_{0},\tilde
{T}_{n_{0}}}^{n_{1},j}\right]  _{\theta=\theta_{0}}=\int\gamma_{\theta
_{0},\tilde{T}_{n_{0}}}^{n_{1},j}\mathrm{tr}\,\partial_{i}\rho_{\theta_{0}%
}^{\otimes n_{0}}\tilde{M}^{n_{0}}\left(  \mathrm{d}\,\vec{\omega}_{1}\right)
$. Therefore, due to Theorem\thinspace VI.2.1 of Holevo\thinspace(1982) and
(\ref{dE}), the first term of (\ref{leibniz-pre}) is evaluated as follows
(they are used to show (i) below).%
\begin{align*}
&  \left\vert \partial_{i}\left(  \mathrm{E}_{\theta}^{\mathsf{\tilde{M}%
}^{n_{0}}}\mathrm{E}_{\theta_{0}}^{\mathsf{M}_{\tilde{T}_{n_{0}}}^{n_{1}}%
}\left[  T_{\tilde{T}_{n_{0}},n_{1}}^{,j}\right]  \right)  _{\theta=\theta
_{0}}\right\vert =\left\vert \partial_{i}\left(  \mathrm{E}_{\theta
}^{\mathsf{\tilde{M}}^{n_{0}}}\left(  \theta_{0}^{j}+\gamma_{\theta_{0}%
,\tilde{T}_{n_{0}}}^{n_{1},j}\right)  \right)  _{\theta=\theta_{0}}\right\vert
=\left\vert \partial_{i}\left[  \mathrm{E}_{\theta}^{\mathsf{\tilde{M}}%
^{n_{0}}}\gamma_{\theta_{0},\tilde{T}_{n_{0}}}^{n_{1},j}\right]
_{\theta=\theta_{0}}\right\vert \\
&  \leq_{\text{(i)}}\sqrt{\mathrm{\mathrm{tr\,}}\rho_{\theta_{0}}^{\otimes
n_{0}}\left(  L_{\theta_{0},i}^{n_{0}}\right)  ^{2}\mathrm{E}_{\theta_{0}%
}^{\mathsf{\tilde{M}}^{n_{0}}}\left(  \gamma_{\theta_{0},\tilde{T}_{n_{0}}%
}^{n_{1},j}\right)  ^{2}}\underset{\text{(ii)}}{\leq}\sqrt{n_{0}%
\mathrm{\mathrm{tr\,}}\rho_{\theta_{0}}\left(  L_{\theta_{0},i}\right)
^{2}\cdot\frac{n_{1}^{2}m^{2}a_{1}a_{4,n_{1}}D_{2,\theta_{0}}}{n_{0}^{2}}},
\end{align*}
\vskip-\lastskip where (ii) is due to (\ref{|gamma|<}) and (M.3.2). Therefore,
the first term vanishes as $n_{0}\rightarrow\infty$. \ \ Due to (\ref{Bxy}) of
Lemma\thinspace\ref{lem:est-cont}, the second term converges to $\partial
_{i}\mathrm{E}_{\theta}^{\mathsf{M}_{\theta_{0}}^{n_{1}}}\left[  T_{\theta
_{0},n_{1}}^{,j}\right]  _{\theta=\theta_{0}}=\delta_{j}^{i}$, and
(\ref{asym-unbiased}) is proved.
\end{proof}

\begin{lemma}
\label{lem:trGV-cont} $\underset{\theta_{0}\rightarrow\theta}{\varlimsup}%
\inf\left\{  \mathrm{Tr}G\,_{\theta}\mathrm{V}_{\theta}\left[  \mathcal{E}%
_{\theta_{0},n}\right]  ;\text{(\ref{locally-unbiased}), (E')}\right\}
=\inf\left\{  \mathrm{Tr}G\,_{\theta}\mathrm{V}_{\theta}\left[  \mathcal{E}%
_{\theta,n}\right]  ;\text{(\ref{locally-unbiased}), (E')}\right\}  $
\end{lemma}

\begin{proof}
Suppose the LHS is larger than the RHS (, denoted by $A$ in the proof) by
$2c>0$. Then one can find a sequence $\left\{  \theta_{k}\right\}  $ such that
$\underset{k\rightarrow\infty}{\lim}\inf\left\{  \mathrm{Tr}G\,_{\theta
}\mathrm{V}_{\theta}\left[  \mathcal{E}_{\theta_{k},n}\right]
;\text{(\ref{locally-unbiased}), (E')}\right\}  =A+2c$. We prove this cannot occur.

Obviously, among those satisfying (\ref{locally-unbiased}), (E'), one can find
$\left\{  \mathcal{E}_{\theta,n}\right\}  _{\theta\in\Theta}$ such that
$\mathrm{Tr}G\,_{\theta}\mathrm{V}_{\theta}\left[  \mathcal{E}_{\theta
,n}\right]  \leq A+c$. Define $\mathcal{E}_{\theta_{k},n}^{\prime}:=\left\{
\mathsf{M}_{\theta}^{n},\,T_{\theta_{k},n}^{\prime}\right\}  $ by
$T_{\theta_{k},n}^{\prime}:=B_{\theta_{k}}\left[  \mathcal{E}_{\theta
,n}\right]  ^{-1}\left(  T_{\theta,n}-\mathrm{E}_{\theta_{k}}^{\mathsf{M}%
_{\theta}^{n}}\left[  T_{\theta,n}\right]  \right)  +\theta_{k}$. It is easy
to verify $\mathrm{V}_{\theta}\left[  \mathcal{E}_{\theta_{k},n}^{^{\prime}%
}\right]  =B_{\theta_{k}}\left[  \mathcal{E}_{\theta,n}\right]  ^{-1}%
\mathrm{Tr}G\,_{\theta}\mathrm{V}_{\theta}\left[  \mathcal{E}_{\theta
,n}\right]  \left(  B_{\theta_{k}}\left[  \mathcal{E}_{\theta,n}\right]
^{-1}\right)  ^{T}$and that $\mathcal{E}_{\theta_{k},n}^{\prime}$ satisfies
(\ref{locally-unbiased}) and (E'). (Here note that $a_{4,n}$ has to be
replaced by the other constant.) Therefore, due to (\ref{Bxy}) of
Lemma\thinspace\ref{lem:est-cont}, \thinspace\thinspace\ \newline%
$\underset{k\rightarrow\infty}{\lim}\mathrm{Tr}G\,_{\theta}\mathrm{V}_{\theta
}\left[  \mathcal{E}_{\theta_{k},n}^{^{\prime}}\right]  =\mathrm{Tr}%
G\,_{\theta}\mathrm{V}_{\theta}\left[  \mathcal{E}_{\theta,n}\right]  \leq
A+c<A+2c=\underset{k\rightarrow\infty}{\lim}\inf\left\{  \mathrm{Tr}%
G\,_{\theta}\mathrm{V}_{\theta}\left[  \mathcal{E}_{\theta_{k},n}\right]
;\text{(\ref{locally-unbiased}), (E')}\right\}  $. This is contradiction.
\end{proof}

Due to Lemmas\thinspace\ref{lem:est-compose}-\ref{lem:trGV-cont}, we have
`$\leq$' of (\ref{cr-q1}) of Theorem\thinspace\ref{th:cr}. \vskip-7mm

\subsection*{\noindent{\protect\normalsize 3.4 On asymptotic normality of the
estimator\thinspace(\ref{def-est})}}

\vskip-3mm The estimator (\ref{def-est}) is asymptotically normal. We prove
the assertion in $m=1$-case, supposing that $\inf_{\theta_{0}\in%
\mathbb{R}
^{m}}\mathrm{V}_{\theta}\left[  \mathcal{E}_{\theta_{0},n_{1}}\right]  $ is
not 0.%
\begin{align*}
&  \left\vert P_{\theta}^{\mathsf{M}^{n_{1},n}}\left\{  \sqrt{n}%
\mathrm{V}_{\theta}\left[  \mathcal{E}_{\theta,n_{1}}\right]  ^{-\frac{1}{2}%
}\left(  T_{n}^{n_{1}}-\theta\right)  \leq y\right\}  -\Phi\left(  y\right)
\right\vert =\left\vert \mathrm{E}_{\theta}^{\mathsf{\tilde{M}}_{n_{0}}%
}\,P_{\theta}^{\mathsf{M}_{\tilde{T}_{n_{0}}}^{n_{1}}}\left\{  \sqrt
{n}\mathrm{V}_{\theta}\left[  \mathcal{E}_{\theta,n_{1}}\right]  ^{-\frac
{1}{2}}\left(  T_{n}^{n_{1}}-\theta\right)  \leq y\right\}  -\Phi\left(
y\right)  \right\vert \\
&  \leq\mathrm{E}_{\theta}^{\mathsf{\tilde{M}}_{n_{0}}}\left\vert \,P_{\theta
}^{\mathsf{M}_{\tilde{T}_{n_{0}}}^{n_{1}}}\left\{  \sqrt{n}\mathrm{V}_{\theta
}\left[  \mathcal{E}_{\tilde{T}_{n_{0}},n_{1}}\right]  ^{-\frac{1}{2}}\left(
T_{n}^{n_{1}}-\theta-\gamma_{\theta,\tilde{T}_{n_{0}}}^{n_{1}}\right)  \leq
y\right\}  -\Phi\left(  y\right)  \right\vert \\
&  +\mathrm{E}_{\theta}^{\mathsf{\tilde{M}}_{n_{0}}}\left\vert \,P_{\theta
}^{\mathsf{M}_{\tilde{T}_{n_{0}}}^{n_{1}}}\left\{  \sqrt{n}\mathrm{V}_{\theta
}\left[  \mathcal{E}_{\theta,n_{1}}\right]  ^{-\frac{1}{2}}\left(
T_{n}^{n_{1}}-\theta\right)  \leq y\right\}  -P_{\theta}^{\mathsf{M}%
_{\tilde{T}_{n_{0}}}^{n_{1}}}\left\{  \sqrt{n}\mathrm{V}_{\theta}\left[
\mathcal{E}_{\tilde{T}_{n_{0}},n_{1}}\right]  ^{-\frac{1}{2}}\left(
T_{n}^{n_{1}}-\theta-\gamma_{\theta,\tilde{T}_{n_{0}}}^{n_{1}}\right)  \leq
y\right\}  \right\vert
\end{align*}%
\vskip-\lastskip
Due to Berry-Esseen bound (Chapter 11 of DasGupta\thinspace(2008)), the first
term is upperbounded by $0.8\left(  a_{4,n_{1}}\right)  ^{3}n_{2}^{-\frac
{1}{2}}\mathrm{E}_{\theta}^{\mathsf{\tilde{M}}_{n_{0}}}\mathrm{V}_{\theta
}\left[  \mathcal{E}_{\tilde{T}_{n_{0}},n_{1}}\right]  ^{-\frac{3}{2}}$ and
converges to 0 as $n_{2}\rightarrow\infty$ since $\inf_{\theta_{0}\in%
\mathbb{R}
^{m}}\mathrm{V}_{\theta}\left[  \mathcal{E}_{\theta_{0},n_{1}}\right]  \neq0$
by assumption. To evaluate the second term, we just have to consider the event
such that $\left\vert \tilde{T}_{n_{0}}-\theta\right\vert <\varepsilon
^{\frac{1}{2}}n^{-\frac{1}{4}}$, since the probability that this does not
occur converges to 0 due to (M.3.2) and Chebyshev's inequality. Due to
Lemma\thinspace\ref{lem:trGV-cont}, we can suppose that $\mathrm{V}_{\theta
}\left[  \mathcal{E}_{\theta_{0},n_{1}}\right]  $ is continuous in $\theta
_{0}$ at $\theta_{0}=\theta$ without loss of generality. Therefore,
$\left\vert \mathrm{V}_{\theta}\left[  \mathcal{E}_{\tilde{T}_{n_{0}},n_{1}%
}\right]  ^{-1/2}-\mathrm{V}_{\theta}\left[  \mathcal{E}_{\theta,n_{1}%
}\right]  ^{-1/2}\right\vert <\varepsilon$ \ for large $n$. Let $c:=y+n_{1}%
^{2}m^{2}a_{1}a_{4,n_{1}}$. Since \newline$\left\vert \frac{1}{\sqrt{n}%
}\left(  \mathrm{V}_{\theta}\left[  \mathcal{E}_{\tilde{T}_{n_{0}},n_{1}%
}\right]  ^{\frac{1}{2}}-\mathrm{V}_{\theta}\left[  \mathcal{E}_{\theta,n_{1}%
}\right]  ^{\frac{1}{2}}\right)  y+\gamma_{\theta,\tilde{T}_{n_{0}}}^{n_{1}%
}\right\vert \leq c\varepsilon n^{-\frac{1}{2}}$ due to (\ref{|gamma|<}), this
is upperbounded by%
\begin{align*}
&  \mathrm{E}_{\theta}^{\mathsf{\tilde{M}}_{n_{0}}}P_{\theta}^{\mathsf{M}%
_{\tilde{T}_{n_{0}}}^{n_{1}}}\left\{  y-c\varepsilon\leq\sqrt{n}%
\mathrm{V}_{\theta}\left[  \mathcal{E}_{\tilde{T}_{n_{0}},n_{1}}\right]
^{-\frac{1}{2}}\left(  T_{n}^{n_{1}}-\theta-\gamma_{\theta,\tilde{T}_{n_{0}}%
}^{n_{1}}\right)  \leq y+c\varepsilon\right\} \\
&  \leq\Phi\left(  y+c\varepsilon\right)  -\Phi\left(  y-c\varepsilon\right)
+0.8\left(  a_{4,n_{1}}\right)  ^{3}n_{2}^{-\frac{1}{2}}\mathrm{E}_{\theta
}^{\mathsf{\tilde{M}}_{n_{0}}}\mathrm{V}_{\theta}\left[  \mathcal{E}%
_{\tilde{T}_{n_{0}},n_{1}}\right]  ^{-\frac{3}{2}}.
\end{align*}%
\vskip-\lastskip
Here the inequality is due to Berry-Esseen bound. Letting $n\rightarrow\infty$
and $\varepsilon\rightarrow0$, the last end converges to 0. After all, we have
our assertion.

In $m\geq2$-case, the first term is evaluated using multi-dimensional version
of Berry-Esseen bound (Chapter 11 of DasGupta\thinspace(2008)). The second
term is evaluated by analogous but more complicated analysis.\quad\vskip-7mm

\subsection*{\noindent{\protect\normalsize 3.5 On logarithmic derivative and
Fisher information}}

\vskip-3mm Hayashi and Matsumoto\thinspace(1998) gives representation of
$C^{Q}\left(  G_{\theta},\mathcal{M}\right)  $ using Fisher information
\ $J_{\theta}^{\mathsf{M}^{n}}$ of the classical statistical model $\left\{
\,P_{\theta}^{\mathsf{M}^{n}}\right\}  _{\theta\in\Theta}$: $C^{Q}\left(
G_{\theta},\mathcal{M}\right)  =\lim_{n\rightarrow\infty}\inf_{\mathsf{M}^{n}%
}n\mathrm{Tr}\,G_{\theta}J_{\theta}^{\mathsf{M}^{n}-1}$.\ They exploits the
fact that the minimum variance of locally unbiased estimators equals $\left(
J_{\theta}^{\mathsf{M}^{n}}\right)  ^{-1}$ and achieved by $T_{n,\theta}%
^{j}=\sum_{i=1}^{m}\left(  J_{\theta}^{\mathsf{M}^{n}-1}\right)
^{ij}l_{\theta,i}^{\mathsf{M}^{n}}+\theta^{j}$, with $l_{\theta,i}%
^{\mathsf{M}^{n}}$'s denoting the logarithmic derivative. Since their
regularity conditions are different from ours, we examine here how far this
statement holds in our setting. First, we define $l_{\theta,i}^{\mathsf{M}%
^{n}}$ as the Radon-Nikodym derivative $\mathrm{d\,tr}\,\partial_{i}%
\rho_{\theta}^{\otimes n}M^{n}/\mathrm{d\,tr}\,\rho_{\theta}^{\otimes n}M^{n}%
$. Let $\mu^{\mathsf{M}^{n}}\left(  \Delta\right)  :=\mathrm{tr}%
\,\sigma^{\otimes n}M^{n}\left(  \Delta\right)  $ ($\sigma>0$) and $p_{\theta
}^{\mathsf{M}^{n}}:=\mathrm{d\,}\left(  \mathrm{tr}\,\rho_{\theta}^{\otimes
n}M^{n}\right)  /\mathrm{d}\mu^{\mathsf{M}^{n}}\,$(, which exists since
$\mathrm{tr}\,\sigma^{\otimes n}M^{n}\left(  \Delta\right)  =0$ implies
$M^{n}\left(  \Delta\right)  =0$). Since $\partial_{i}\int f\,p_{\theta
}^{\mathsf{M}^{n}}\mathrm{d}\mu^{\mathsf{M}^{n}}\leq\left(  \sup\left\vert
f\right\vert \right)  \left\Vert \partial_{i}\rho_{\theta}^{\otimes
n}\right\Vert _{1}$, there is $L^{1}$ function $\partial_{i}p_{\theta
}^{\mathsf{M}^{n}}$ such that $\partial_{i}\int f\,p_{\theta}^{\mathsf{M}^{n}%
}\mathrm{d}\mu^{\mathsf{M}^{n}}=\int f\,\partial_{i}p_{\theta}^{\mathsf{M}%
^{n}}\mathrm{d}\mu^{\mathsf{M}^{n}}$. Using this, $l_{\theta,i}^{\mathsf{M}%
^{n}}=\partial_{i}p_{\theta}^{\mathsf{M}^{n}}/p_{\theta}^{\mathsf{M}^{n}}$, if
the RHS is finite.

\begin{lemma}
\label{lemfisher-finite}Suppose that (M.1) and (M.2) holds. Then $l_{\theta
,i}^{\mathsf{M}^{n}}$ exists and $\sup_{\mathsf{M}^{n}}\mathbf{v}^{T}%
J_{\theta}^{\mathsf{M}^{n}}\mathbf{v}J_{\theta}^{\mathsf{M}^{n}}$ is finite.
Also, if $\partial_{i}\rho_{\theta}\in\tau c\left(  \mathcal{H}\right)  $
exists, and $\rho_{\theta}>0$, $l_{\theta,i}^{\mathsf{M}^{n}}$ exists.
\end{lemma}

\begin{proof}
Suppose (M.2) holds. Since $\rho_{\theta}\geq0$ and $M\left(  \cdot\right)
\geq0$, $\mathrm{tr}\,\rho_{\theta}^{\otimes n}M=0$ means $\rho_{\theta
}^{\otimes n}M=0$. Therefore, \ due to (M.2), $\mathrm{tr}\,\partial_{i}%
\rho_{\theta}^{\otimes n}M=\frac{1}{2}$ $\left(  \mathrm{tr}\,L_{\theta,i}%
^{n}\rho_{\theta}^{\otimes n}M+\mathrm{tr}\,M\rho_{\theta}^{\otimes
n}L_{\theta,i}^{n}\right)  =0$. Therefore, $l_{\theta,i}^{\mathsf{M}^{n}}$
exists. Define $l_{\theta,\mathbf{v}}^{\mathsf{M}^{n}}:=\sum_{i=1}^{m}%
v_{i}l_{\theta,i}^{\mathsf{M}^{n}}$. Let $\Delta_{\iota}:=\left\{
\omega;\iota\varepsilon\leq\left(  l_{\theta,\mathbf{v}}^{\mathsf{M}^{n}%
}\left(  \omega\right)  \right)  ^{2}\leq\left(  \iota+1\right)
\varepsilon\right\}  $, and denote by $\omega_{\iota}$ \ the one satisfying
$\left(  l_{\theta,\mathbf{v}}^{\mathsf{M}^{n}}\left(  \omega\right)  \right)
^{2}=\iota\varepsilon$. Observe
\begin{align*}
0  &  \leq\sum_{i,j}v_{i}v_{j}\sum_{\iota}\mathrm{tr}\,\rho_{\theta}^{\otimes
n}\left\{  L_{\theta,i}^{n}-l_{\theta,i}^{\mathsf{M}^{n}}\left(  \omega
_{\iota}\right)  \right\}  M^{n}\left(  \Delta_{\iota}\right)  \left\{
L_{\theta,j}^{n}-l_{\theta,j}^{\mathsf{M}^{n}}\left(  \omega_{\iota}\right)
\right\} \\
&  =\mathbf{v}^{T}J_{\theta}^{S,n}\mathbf{v-}2\sum_{j}\mathbf{\sum_{\iota}%
}v_{j}l_{\theta,\mathbf{v}}^{\mathsf{M}^{n}}\left(  \omega_{\iota}\right)
\Re\mathrm{tr}\,\rho_{\theta}^{\otimes n}L_{\theta,j}^{n}M^{n}\left(
\Delta_{\iota}\right)  +\mathbf{\sum_{\iota}}\left(  l_{\theta,\mathbf{v}%
}^{\mathsf{M}^{n}}\left(  \omega_{\iota}\right)  \right)  ^{2}\mathrm{tr}%
\,\rho_{\theta}^{\otimes n}M^{n}\left(  \Delta_{\iota}\right) \\
&  =\mathbf{v}^{T}J_{\theta}^{S,n}\mathbf{v-}2\mathbf{\sum_{\iota}}%
l_{\theta,\mathbf{v}}^{\mathsf{M}^{n}}\left(  \omega_{\iota}\right)
\int_{\Delta_{\iota}}l_{\theta,\mathbf{v}}^{\mathsf{M}^{n}}\left(
\omega\right)  \mathrm{tr}\,\rho_{\theta}^{\otimes n}M^{n}\left(
\mathrm{d}\omega\right)  +\mathbf{\sum_{\iota}}\left(  l_{\theta,\mathbf{v}%
}^{\mathsf{M}^{n}}\left(  \omega_{\iota}\right)  \right)  ^{2}\mathrm{tr}%
\,\rho_{\theta}^{\otimes n}M^{n}\left(  \Delta_{\iota}\right) \\
&  =\mathbf{v}^{T}J_{\theta}^{S,n}\mathbf{v}-\mathbf{\sum_{\iota}}\left(
l_{\theta,\mathbf{v}}^{\mathsf{M}^{n}}\left(  \omega_{\iota}\right)  \right)
^{2}\mathrm{tr}\,\rho_{\theta}^{\otimes n}M^{n}\left(  \Delta_{\iota}\right)
-2\mathbf{\sum_{\iota}}l_{\theta,\mathbf{v}}^{\mathsf{M}^{n}}\left(
\omega_{\iota}\right)  \left(  \int_{\Delta_{\iota}}l_{\theta,\mathbf{v}%
}^{\mathsf{M}^{n}}\left(  \omega\right)  \mathrm{tr}\,\rho_{\theta}^{\otimes
n}M^{n}\left(  \mathrm{d}\omega\right)  -l_{\theta,\mathbf{v}}^{\mathsf{M}%
^{n}}\left(  \omega_{\iota}\right)  \mathrm{tr}\,\rho_{\theta}^{\otimes
n}M^{n}\left(  \Delta_{\iota}\right)  \right) \\
&  \leq\mathbf{v}^{T}J_{\theta}^{S,n}\mathbf{v}-\int\left(  l_{\theta
,\mathbf{v}}^{\mathsf{M}^{n}}\left(  \omega\right)  \right)  ^{2}\rho_{\theta
}^{\otimes n}M^{n}\left(  \mathrm{d}\omega\right)  +\varepsilon+2\mathbf{\sum
_{\iota}}\int_{\Delta_{\iota}}\left\vert l_{\theta,\mathbf{v}}^{\mathsf{M}%
^{n}}\left(  \omega_{\iota}\right)  \right\vert \frac{\varepsilon}{\left\vert
l_{\theta,\mathbf{v}}^{\mathsf{M}^{n}}\left(  \omega\right)  +l_{\theta
,\mathbf{v}}^{\mathsf{M}^{n}}\left(  \omega_{\iota}\right)  \right\vert
}\mathrm{tr}\,\rho_{\theta}^{\otimes n}M^{n}\left(  \mathrm{d}\omega\right) \\
&  \leq\mathbf{v}^{T}J_{\theta}^{S,n}\mathbf{v}-\int\left(  l_{\theta
,\mathbf{v}}^{\mathsf{M}^{n}}\left(  \omega\right)  \right)  ^{2}\rho_{\theta
}^{\otimes n}M^{n}\left(  \mathrm{d}\omega\right)  +\varepsilon+2\varepsilon
\end{align*}%
\vskip-\lastskip
which, with $\varepsilon\rightarrow0$, implies $\mathbf{v}^{T}J_{\theta
}^{\mathsf{M}^{n}}\mathbf{v}\leq\mathbf{v}^{T}J_{\theta}^{S,n}\mathbf{v}%
<\infty$. Also, suppose $\rho_{\theta}>0$. Then, $\mathrm{tr}\,\rho_{\theta
}^{\otimes n}M=0$ means $M=0$ and $\mathrm{tr}\,\partial_{i}\rho_{\theta
}^{\otimes n}M=0$. Therefore, the second assertion is proved.
\end{proof}

The RHS of (\ref{cr-q1}) and (\ref{cr-q2}) is lowerbounded by $\inf
_{\mathsf{M}^{n}}\mathrm{Tr}\,G_{\theta}\left(  J_{\theta}^{\mathsf{M}^{n}%
}\right)  ^{-1}$, due to Schwartz's inequality. Achievability, in fact, also
holds. Define $\mathcal{E}_{\theta,n}^{L}:=\left\{  \mathsf{M}_{\theta}%
^{n},\,T_{n,\theta}^{L}\right\}  $ by $l_{\theta,i}^{\mathsf{M}^{n}%
\mathsf{,}L}:=\chi_{\left\{  \omega_{n};\left\Vert l_{\theta,i}^{\mathsf{M}%
}\right\Vert \leq L\right\}  }l_{\theta,i}^{\mathsf{M}^{n}}$, $\left(
J_{\theta}^{\mathsf{M}^{n},L}\right)  _{i,j}:=\mathrm{\,E}_{\theta
}^{\mathsf{M}^{n}}\,l_{\theta,i}^{\mathsf{M}^{n}\mathsf{,}L}l_{\theta
,j}^{\mathsf{M}^{n}\mathsf{,}L}$, and $T_{n,\theta}^{L,j}:=\sum_{i=1}%
^{m}\left[  \left(  J_{\theta}^{\mathsf{M}^{n},L}\right)  ^{-1}\right]
^{ij}l_{\theta,i}^{\mathsf{M}^{n},L}+\theta^{j}$. Obviously, $\mathcal{E}%
_{\theta,n}^{L}$ satisfies (E'). Therefore, due to Lemma\thinspace
\ref{lem:est-cont}, \newline$\quad\partial_{j}\mathrm{E}_{\theta}%
^{\mathsf{M}^{n}}\,l_{\theta_{0},i}^{\mathsf{M}^{n}\mathsf{,}L}=\,\int
l_{\theta,i}^{\mathsf{M}^{n}\mathsf{,}L}\mathrm{tr}\,\partial_{j}\rho_{\theta
}^{\otimes n}M^{n}\left(  \mathrm{d}\omega_{n}\right)  =\int l_{\theta
,i}^{\mathsf{M}^{n}\mathsf{,}L}l_{\theta,j}^{\mathsf{M}^{n}}\mathrm{tr}%
\,\rho_{\theta}^{\otimes n}M^{n}\left(  \mathrm{d}\omega_{n}\right)  =\int
l_{\theta,i}^{\mathsf{M}^{n}\mathsf{,}L}l_{\theta,j}^{\mathsf{M}^{n}%
\mathsf{,}L}\mathrm{tr}\,\rho_{\theta}M^{n}\left(  \mathrm{d}\omega
_{n}\right)  =J_{\theta}^{\mathsf{M}^{n},L}$.

\noindent Therefore, $\mathcal{E}_{\theta,n}^{L}$ also satisfies
(\ref{asym-unbiased}). Also, $\mathrm{Tr}\,G_{\theta}\mathrm{V}_{\theta
}\left[  \mathcal{E}_{\theta,n}^{L}\right]  =\,\mathrm{Tr}\,G_{\theta}\left(
J_{\theta}^{\mathsf{M}^{n},L}\right)  ^{-1}$. Hence, it remains to show
$\underset{L\rightarrow\infty}{\lim}\mathbf{v}^{T}J_{\theta}^{\mathsf{M}%
^{n},L}\mathbf{v}=\mathbf{v}^{T}J_{\theta}^{\mathsf{M}^{n}}\mathbf{v}$,
$\forall\mathbf{v}$. This is true since $\mathbf{v}^{T}J_{\theta}%
^{\mathsf{M}^{n}}\mathbf{v}=\int\left(  \sum_{i=1}^{m}v^{i}l_{\theta
,i}^{\mathsf{M}^{n}}\right)  ^{2}P_{\theta}^{\mathsf{M}^{n}}\left(
\mathrm{d}\omega_{n}\right)  <\infty$.

Therefore, logarithmic derivative and the Fisher information can be used to
represent $C^{Q}\left(  G_{\theta},\mathcal{M}\right)  $. However, it is not
possible to show their chain rule, which is at the heart of the argument for
one-way semi-classical setting in Hayashi and Matsumoto\thinspace(1998).
Therefore, in the next subsection, we use somewhat different method to prove
the asymptotic Cramer-Rao type bound in the semi-classical setting.

\vskip-7mm

\section{\noindent{\protect\normalsize 4. SEMI-CLASSICAL MEASUREMENT}}

\vskip-3mm

\subsection*{\noindent{\protect\normalsize 4.1 DEFINITIONS, REGULARITY
CONDITIONS, AND\ MAIN THEOREM}}

\vskip-3mm An important subclass of measurements is \textit{semi-classical
}measurements, which are composed adoptively in $R_{n}$ (\thinspace$<\infty$)
rounds. At each round, we measure each sample separately, and the measurements
of the $r$-th round depend on the previously obtained data. We denote by
$z_{r,\kappa}$($\in$ $%
\mathbb{R}
^{l}$) the data obtained at the $r$-th round from $\kappa$-th sample, and
$\mathbf{z}_{r}$ is the data $\left(  z_{1,1},z_{1,2},\cdots,z_{r,n}\right)  $
obtained up to the $r$-th round. The measurement acting in the $r$-th round on
$\kappa$-th sample is denoted by $\mathsf{M}_{r,\kappa}^{\mathbf{z}_{r-1}}$.
Without loss of generality, we suppose that in the first round the measurement
is chosen deterministically. Rigorous mathematical description of such a
process is given in the following subsection.

Define $\mathbf{z}_{r,\kappa}^{\downarrow}:=$ $\left(  \mathbf{z}%
_{r-1},z_{r,1},\cdots z_{r,\kappa}\right)  $, and \ $\mathbf{z}_{r,\kappa
}^{\uparrow}:=$ $\left(  z_{r,\kappa},z_{r,\kappa+1},\cdots,z_{R_{n}%
,n}\right)  $. $\mathfrak{B}_{r}$, $\mathfrak{B}_{r,\kappa}$, $\mathfrak{B}%
_{r,\kappa}^{\downarrow}$ and $\mathfrak{B}_{r,\kappa}^{\uparrow}$ is the
totality of Borel sets over the space where $\mathbf{z}_{r}$ , $z_{r,\kappa}$,
$\mathbf{z}_{r,\kappa}^{\downarrow}$, and $\mathbf{z}_{r,\kappa}^{\uparrow}$
is living in, respectively. The instrument corresponding to successive
application of $\mathsf{M}_{1,1}$, $\mathsf{M}_{1,2}$, $\cdots$,
$\mathsf{M}_{r,\kappa}^{\mathbf{z}_{r-1}}$ and $\mathsf{M}_{r,\kappa
}^{\mathbf{z}_{r-1}}$, $\mathsf{M}_{r,\kappa+1}^{\mathbf{z}_{r-1}}$, $\cdots$,
$\mathsf{M}_{R_{n},n}^{\mathbf{z}_{R_{n}-1}}$ is denoted by $\mathsf{M}%
_{r,\kappa}^{\downarrow}$ and $\mathsf{M}_{r,\kappa}^{\uparrow\,\mathbf{z}%
_{r-1}}$, respectively. Note that they depend on $n$, although we do not
denote the fact explicitly for the sake of simplicity. Note also that $R_{n}$
is arbitrary but finite.

Note that in other literatures such as Hayashi and Matsumoto\thinspace(1999),
the term `semi-classical measurement' refers to more restricted class of
measurement, which is called \textit{one-way semi-classical} measurement in
this paper. The restriction is that in $r$-th round, we measure $r$-th sample
only (Hence, $R_{n}=n$).

The \textit{asymptotic semi-classical Cramer-Rao} type bound $C_{\theta
}\left(  G_{\theta},\mathcal{M}\right)  $ is defined by \newline$\quad
\quad\quad C_{\theta}\left(  G_{\theta},\mathcal{M}\right)  :=\underset
{n\rightarrow\infty}{\varlimsup}\inf\left\{  n\mathrm{Tr}\,G_{\theta
}\mathrm{MSE}_{\theta}\left[  \mathcal{E}_{n}\right]  \,;\mathsf{M}^{n}\text{
in }\mathcal{H}^{\otimes n}\text{, semi-classical, (\ref{asym-unbiased}),
(E)}\right\}  .$\newline

\begin{theorem}
\label{th:cr-cl}Suppose (M.1-3) hold. Then, \ \newline%
$\ \ \ \ \ \ \ \ \ \ \ \ \ C_{\theta}\left(  G_{\theta},\mathcal{M}\right)
=\inf\left\{  \mathrm{Tr}\,G_{\theta}\mathrm{V}_{\theta}\left[  \mathcal{E}%
_{\theta_{0},1}\right]  \,\text{\thinspace};\text{ }\mathsf{M}^{1}\text{in
}\mathcal{H}\text{, (\ref{locally-unbiased}), (E') with }n=1\right\}  ,$

$\ \ \ \ \ \ \ \ \ \ \ \ \ \ \ \ \ \ \ \ \ \ \ =\inf\left\{  \mathrm{Tr}%
\,G_{\theta}\mathrm{V}_{\theta}\left[  \mathcal{E}_{\theta_{0},1}\right]
\,\text{\thinspace};\text{ }\mathsf{M}^{1}\text{in }\mathcal{H}\text{,
(\ref{locally-unbiased}), (E) with }n=1\right\}  .$
\end{theorem}

\vskip-5mm

\subsection*{{\protect\normalsize \noindent4.2 ADAPTIVE MEASUREMENT}}

\vskip-3mm In this subsection, we give mathematically rigorous account on a
composite measurement $\mathsf{NM}$ of measurement $\mathsf{M}$ followed by
$\mathsf{N}^{\omega}$, where $\mathsf{N}^{\omega}$ is composed depending on
the data $\omega\in%
\mathbb{R}
^{l}$ from $\mathsf{M}$. More specifically, $\omega\rightarrow\mathsf{N}%
^{\omega}\left[  \Delta^{\prime}\right]  $ can be approximated by a sequence
of simple functions except for $\omega\in\Delta$ where $M\left(
\Delta\right)  =0$, so that the function is strongly measurable with respect
to $P_{\rho}^{\mathsf{M}}\mathrm{\,}$\ for any $\rho$. We show this composite
$\mathsf{NM}$ can be described using an instrument. (The contents of this
subsection should be well-known to specialists of the field of measurement
theory. The author, however, could not find a proper reference.)

The key fact is Theorem\thinspace4.5 of Ozawa\thinspace(1985), or that there
is a family of density operators $\left\{  \rho_{\omega}^{\mathsf{M}}\right\}
_{\omega\in%
\mathbb{R}
^{l}}$ (\textit{a posteriori states}) with $\int_{\omega\in\Delta}%
\mathrm{tr}\,A\rho_{\omega}^{\mathsf{M}}P_{\rho}^{\mathsf{M}}\left(
\mathrm{\,d}\,\omega\right)  =$ $\mathrm{tr}\,A\,\mathsf{M}\left[
\Delta\right]  \left(  \rho\right)  $ ($\forall A\in\mathcal{B}\left(
\mathcal{H}\right)  $). Since $\mathcal{B}\left(  \mathcal{H}\right)  $ is the
dual of $\mathcal{\tau}c\left(  \mathcal{H}\right)  $ with the pairing
$\left\langle \rho,A\right\rangle :=\mathrm{tr}\,\rho A$ (Theorem\thinspace
II.7.2 of Holevo\thinspace(1982)), Ozawa's statement is equivalent to the weak
measurability and the existence of Pettis integral of the function
$\omega\rightarrow\rho_{\omega}^{\mathsf{M}}$. As summarized in the end of
Subsection\thinspace2.1, $\omega\rightarrow\rho_{\omega}^{\mathsf{M}}$ in fact
is strongly measurable. Also, since $\int_{\omega\in\Delta}\left\Vert
\rho_{\omega}^{\mathsf{M}}\right\Vert _{1}P_{\rho}^{\mathsf{M}}\left(
\mathrm{\,d}\,\omega\right)  =1<\infty$, the Bochner integral $\int_{\omega
\in\Delta}\rho_{\omega}^{\mathsf{M}}P_{\rho}^{\mathsf{M}}\left(
\mathrm{\,d}\,\omega\right)  =$ $\mathsf{M}\left[  \Delta\right]  \left(
\rho\right)  $ is convergent.

First, we show $P_{\rho}^{\mathsf{NM}}$ is well-defined. Since \ $\omega
\rightarrow\mathsf{N}^{\omega}\left[  \Delta^{\prime}\right]  $ and
$\omega\rightarrow\rho_{\omega}^{\mathsf{M}}$ are strongly measurable, they
can be approximated by simple functions. Therefore, $\omega\rightarrow
\mathrm{tr}\,\mathsf{N}^{\omega}\left[  \Delta^{\prime}\right]  \rho_{\omega
}^{\mathsf{M}}$ is a measurable function for any $\Delta^{\prime}%
\in\mathfrak{B}\left(
\mathbb{R}
^{l}\right)  $, and $P_{\rho}^{\mathsf{NM}}\left(  \Delta\times\Delta^{\prime
}\right)  :=\int_{\omega\in\Delta}$ $\mathrm{tr}\,\mathsf{N}^{\omega}\left[
\Delta^{\prime}\right]  \rho_{\omega}^{\mathsf{M}}P_{\rho}^{\mathsf{M}}\left(
\mathrm{\,d}\,\omega\right)  $ is well-defined and $\sigma$-additive.
Therefore, $P_{\rho}^{\mathsf{NM}}$ can be extended to $\mathfrak{B}\left(
\mathbb{R}
^{l}\times%
\mathbb{R}
^{l}\right)  $ due to Hopf's extension theorem. Moreover, with $\tilde{\Delta
}_{\omega}:=\{\omega^{\prime};\left(  \omega,\omega^{\prime}\right)  \in
\tilde{\Delta}\}$, $\int_{\omega}\mathrm{tr}\,\mathsf{N}^{\omega}\left[
\tilde{\Delta}_{\omega}\right]  \rho_{\omega}^{\mathsf{M}}\,P_{\rho
}^{\mathsf{M}}\left(  \mathrm{\,d}\,\omega\right)  $ exists and equals
$P_{\rho}^{\mathsf{NM}}\left\{  \tilde{\Delta}\right\}  $ for any Borel set
$\tilde{\Delta}$; Let $\mathfrak{D}$ be the totality of $\tilde{\Delta}$ such
that the assertion is true. Obviously, $\mathfrak{D}$ is a Dynkin system, and
contains cylinder sets. Therefore, due to Dynkin's lemma, $\mathfrak{D=B}%
\left(
\mathbb{R}
^{l}\times%
\mathbb{R}
^{l}\right)  $.

Next, we show that $\rho_{\tilde{\Delta}}^{\mathsf{NM}}$ is well-defined.
Since \ $\int_{\omega}\left\Vert \mathsf{N}^{\omega}\left[  \tilde{\Delta
}_{\omega}\right]  \rho_{\omega}^{\mathsf{M}}\right\Vert _{1}P_{\rho
}^{\mathsf{M}}\left(  \mathrm{\,d}\,\omega\right)  \leq1$, the Bochner
integral $\mathsf{NM}\left[  \tilde{\Delta}\right]  \left(  \rho\right)
:=\int_{\omega}\mathsf{N}^{\omega}\left[  \tilde{\Delta}_{\omega}\right]
\rho_{\omega}^{\mathsf{M}}\,P_{\rho}^{\mathsf{M}}\left(  \mathrm{\,d}%
\,\omega\right)  $ is convergent. Also, its trace equals $P_{\rho
}^{\mathsf{NM}}\left\{  \tilde{\Delta}\right\}  $, since $\mathrm{tr}\,$\ and
$\int$ can be exchanged due to Fubini's theorem. $\mathsf{\ }$

In addition, $\rho\rightarrow\mathsf{NM}\left[  \tilde{\Delta}\right]  \left(
\rho\right)  $ is affine and completely positive, as is proved in the
following. Observe Bochner integral $\int_{\omega}\mathsf{N}^{\omega}\left[
\tilde{\Delta}_{\omega}\right]  P_{\rho}^{\mathsf{M}}\left(  \mathrm{\,d}%
\,\omega\right)  $ in $\mathcal{B}\left(  \tau c\left(  \mathcal{H}\right)
\right)  $ is well-defined in terms of $\left\Vert \cdot\right\Vert _{cb}$,
due to \newline$\int_{\omega}\left\Vert \mathsf{N}^{\omega}\left[
\tilde{\Delta}_{\omega}\right]  \right\Vert _{cb}P_{\rho}^{\mathsf{M}}\left(
\mathrm{\,d}\,\omega\right)  \leq1$. Therefore, there exist sequences of
families $\left\{  \mathsf{N}_{j}^{\left(  k\right)  }\right\}  _{j}$ and
$\left\{  \Delta_{j}^{\left(  k\right)  }\right\}  _{j}$ ($k=1$,$\cdots$,
$\infty$) of completely positive maps and Borel sets, such that for any
$\rho$
\begin{align*}
&  \left\Vert \mathsf{NM}\left[  \tilde{\Delta}\right]  \left(  \rho\right)
-\sum_{j}\mathsf{N}_{j}^{\left(  k\right)  }\mathsf{M}\left[  \Delta
_{j}^{\left(  k\right)  }\right]  \left(  \rho\right)  \right\Vert
_{1}=\left\Vert \int_{\omega}\left\{  \,\mathsf{N}^{\omega}\left[
\tilde{\Delta}_{\omega}\right]  -\sum_{j}\mathsf{N}_{j}^{\left(  k\right)
}\chi_{\Delta_{j}^{\left(  k\right)  }}\right\}  \rho_{\omega}^{\mathsf{M}%
}P_{\rho}^{\mathsf{M}}\left(  \mathrm{\,d}\,\omega\right)  \right\Vert _{1}\\
&  \leq\int_{\omega}\,\left\Vert \left\{  \mathsf{N}^{\omega}\left[
\tilde{\Delta}_{\omega}\right]  -\sum_{j}\mathsf{N}_{j}^{\left(  k\right)
}\chi_{\Delta_{j}^{\left(  k\right)  }}\right\}  \rho_{\omega}^{\mathsf{M}%
}\right\Vert _{1}P_{\rho}^{\mathsf{M}}\left(  \mathrm{\,d}\,\omega\right)
\leq\int_{\omega}\,\left\Vert \mathsf{N}^{\omega}\left[  \tilde{\Delta
}_{\omega}\right]  -\sum_{j}\mathsf{N}_{j}^{\left(  k\right)  }\chi
_{\Delta_{j}^{\left(  k\right)  }}\right\Vert _{cb}P_{\rho}^{\mathsf{M}%
}\left(  \mathrm{\,d}\,\omega\right)  \rightarrow0,
\end{align*}%
\vskip-\lastskip
as $k\rightarrow\infty$. Since $\sum_{j}\mathsf{N}_{j}^{\left(  k\right)
}\mathsf{M}\left[  \Delta_{j}^{\left(  k\right)  }\right]  $ is affine and
completely positive, we have our assertion.

Finally, $\tilde{\Delta}\rightarrow\mathsf{NM}\left[  \tilde{\Delta}\right]  $
is an instrument. Obviously, $\mathrm{\mathrm{tr}\,}\mathsf{NM}\left[
\mathbb{R}
^{l}\times%
\mathbb{R}
^{l}\right]  \left(  \rho\right)  =1$. Also,
\begin{align*}
\mathsf{NM}\left[  \bigcup_{j=1}^{\infty}\tilde{\Delta}_{j}\right]  \left(
\rho\right)   &  =\int_{\omega}\mathsf{N}^{\omega}\left[  \bigcup
_{j=1}^{\infty}\tilde{\Delta}_{j,\omega}\right]  \rho_{\omega}^{\mathsf{M}%
}P_{\rho}^{\mathsf{M}}\left(  \mathrm{\,d}\,\omega\right)  =\int_{\omega}%
\sum_{j=1}^{\infty}\mathsf{N}^{\omega}\left[  \tilde{\Delta}_{j,\omega
}\right]  \rho_{\omega}^{\mathsf{M}}P_{\rho}^{\mathsf{M}}\left(
\mathrm{\,d}\,\omega\right) \\
&  =\sum_{j=1}^{\infty}\int_{\omega}\mathsf{N}^{\omega}\left[  \tilde{\Delta
}_{j,\omega}\right]  \rho_{\omega}^{\mathsf{M}}P_{\rho}^{\mathsf{M}}\left(
\mathrm{\,d}\,\omega\right)  =\sum_{j=1}^{\infty}\mathsf{NM}\left[
\tilde{\Delta}_{j}\right]  \left(  \rho\right)  ,
\end{align*}%
\vskip-\lastskip
where the third identity is due to Fubini's theorem of Bochner integral.
Therefore, $\tilde{\Delta}\rightarrow\mathsf{NM}\left[  \tilde{\Delta}\right]
$ is $\sigma$-additive in terms of strong operator topology in $\mathcal{B}%
\left(  \tau c\left(  \mathcal{H}\right)  \right)  $. \vskip-3mm

\subsection*{{\protect\normalsize \noindent4.3 LEIBNIZ RULE}}

\vskip-3mm In this subsection and the next, so far as no confusion is likely
to arise, we denote $P_{\theta}^{\mathsf{M}_{\theta}^{n}}$ and $\mathrm{E}%
_{\theta}^{\mathsf{M}_{\theta}^{n}}$ by $P_{\theta}$ and $\mathrm{E}_{\theta}%
$, respectively, where $\mathsf{M}_{\theta}^{n}$ is a semi-classical measurement.

\begin{lemma}
\label{lem:ec-weak}Suppose $T_{n}$ satisfies (E'). Suppose \ also (a)
$\rho_{\theta}>0$, $\forall\theta\in\Theta$, or (b) the estimator is one-way
semi-classical. Then, $\exists\Delta\in\mathfrak{B}_{n}$ s.t. $M^{n}%
(\Delta\mathsf{)=0}$ and $\mathrm{E}_{\theta}\left[  T_{n}|\,\mathfrak{B}%
_{r,\kappa}^{\downarrow}\right]  \left(  \mathbf{z}_{r,\kappa}^{\downarrow
}\right)  $ is continuous in $\theta$ for $\forall r$ $\forall\kappa$, if
$\mathbf{z}_{n}\notin\Delta$.
\end{lemma}

\begin{proof}
The case (b) is due to the fact that $\mathsf{M}_{r,r}^{\uparrow
\mathbf{z}_{r-1}}$ acts on $\rho_{\theta}^{\otimes\left(  n-r+1\right)  }$:
\newline$\mathrm{E}_{\theta}\left[  T_{n}|\,\mathfrak{B}_{r,\kappa
}^{\downarrow}\right]  \left(  \mathbf{z}_{r,\kappa}^{\downarrow}\right)
=\int T_{n}\left(  \mathbf{z}_{r,r}^{\uparrow},\mathbf{z}_{r-1}\right)
\mathrm{tr}\,\rho_{\theta}^{\otimes\left(  n-r+1\right)  }M_{r,r}%
^{\uparrow\mathbf{z}_{r-1}}\left(  \mathrm{d}\mathbf{z}_{r,r}^{\uparrow
}\right)  $. Here, by abuse of notation, $M_{r,r}^{\uparrow\mathbf{z}_{r-1}%
}\left(  \Delta\right)  \in\mathcal{B}\left(  \mathcal{H}^{\otimes
n-r+1}\right)  $ in case of one-way semi-classical measurements.

The case (a) is proved as follows. \ Let $g_{\theta,\theta^{\prime}}\left(
\mathbf{z}_{r,\kappa}^{\downarrow}\right)  :=\mathrm{E}_{\theta^{\prime}%
}\left[  T_{n}|\,\mathfrak{B}_{r,\kappa}^{\downarrow}\right]  \left(
\mathbf{z}_{r,\kappa}^{\downarrow}\right)  -\mathrm{E}_{\theta}\left[
T_{n}|\,\mathfrak{B}_{r,\kappa}^{\downarrow}\right]  \left(  \mathbf{z}%
_{r,\kappa}^{\downarrow}\right)  $. Observe \ $\int_{\mathbf{z}_{r,\kappa
}^{\downarrow}\in\Delta}\mathrm{E}_{\theta}\left[  T_{n}|\,\mathfrak{B}%
_{r,\kappa}^{\downarrow}\right]  \left(  \mathbf{z}_{r,\kappa}^{\downarrow
}\right)  P_{\theta}^{\mathsf{M}_{r,\kappa}^{\downarrow}}\left(
\mathrm{d}\mathbf{z}_{r,\kappa}^{\downarrow}\right)  $, which equals\newline%
$\int_{\mathbf{z}_{r,\kappa}^{\downarrow}\in\Delta}T_{n}\left(  \mathbf{z}%
_{r,\kappa}^{\downarrow},\mathbf{z}_{r,\kappa+1}^{\uparrow}\right)
\mathrm{\mathrm{tr}\,}\rho_{\theta}^{\otimes n}M^{n}\left(  \mathrm{d}%
\mathbf{z}_{r,\kappa}^{\downarrow}\mathrm{d}\mathbf{z}_{r,\kappa+1}^{\uparrow
}\right)  $, is continuous in $\theta$ for any Borel set $\Delta$, due to
(\ref{E-E'}) of Lemma\thinspace\ref{lem:est-cont}. Also, observe
\begin{align*}
&  \int_{\mathbf{z}_{r,\kappa}^{\downarrow}\in\Delta}g_{\theta,\theta^{\prime
}}\left(  \mathbf{z}_{r,\kappa}^{\downarrow}\right)  P_{\theta}^{\mathsf{M}%
_{r,\kappa}^{\downarrow}}\left(  \mathrm{d}\mathbf{z}_{r,\kappa}^{\downarrow
}\right) \\
&  =\int_{\mathbf{z}_{r,\kappa}^{\downarrow}\in\Delta}\left\{  \mathrm{E}%
_{\theta^{\prime}}\left[  T_{n}|\,\mathfrak{B}_{r,\kappa}^{\downarrow}\right]
\left(  \mathbf{z}_{r,\kappa}^{\downarrow}\right)  P_{\theta^{\prime}%
}^{\mathsf{M}_{r,\kappa}^{\downarrow}}\left(  \mathrm{d}\mathbf{z}_{r,\kappa
}^{\downarrow}\right)  -\mathrm{E}_{\theta}\left[  T_{n}|\,\mathfrak{B}%
_{r,\kappa}^{\downarrow}\right]  \left(  \mathbf{z}_{r,\kappa}^{\downarrow
}\right)  P_{\theta}^{\mathsf{M}_{r,\kappa}^{\downarrow}}\left(
\mathrm{d}\mathbf{z}_{r,\kappa}^{\downarrow}\right)  \right\} \\
&  +\int_{\mathbf{z}_{r,\kappa}^{\downarrow}\in\Delta}\mathrm{E}%
_{\theta^{\prime}}\left[  T_{n}|\,\mathfrak{B}_{r,\kappa}^{\downarrow}\right]
\left(  \mathbf{z}_{r,\kappa}^{\downarrow}\right)  \left\{  P_{\theta
}^{\mathsf{M}_{r,\kappa}^{\downarrow}}\left(  \mathrm{d}\mathbf{z}_{r,\kappa
}^{\downarrow}\right)  -P_{\theta^{\prime}}^{\mathsf{M}_{r,\kappa}%
^{\downarrow}}\left(  \mathrm{d}\mathbf{z}_{r,\kappa}^{\downarrow}\right)
\right\}  .\\
&  =\int_{\mathbf{z}_{r,\kappa}^{\downarrow}\in\Delta}T_{n}\left(
\mathbf{z}_{r,\kappa}^{\downarrow},\mathbf{z}_{r,\kappa+1}^{\uparrow}\right)
\mathrm{\mathrm{tr}\,}\left(  \rho_{\theta^{\prime}}^{\otimes n}-\rho_{\theta
}^{\otimes n}\right)  M^{n}\left(  \mathrm{d}\mathbf{z}_{r,\kappa}%
^{\downarrow}\mathrm{d}\mathbf{z}_{r,\kappa+1}^{\uparrow}\right) \\
&  +\int_{\mathbf{z}_{r,\kappa}^{\downarrow}\in\Delta}\mathrm{E}%
_{\theta^{\prime}}\left[  T_{n}|\,\mathfrak{B}_{r,\kappa}^{\downarrow}\right]
\left(  \mathbf{z}_{r,\kappa}^{\downarrow}\right)  \mathrm{\mathrm{tr}%
\,}\left(  \rho_{\theta^{\prime}}^{\otimes n}-\rho_{\theta}^{\otimes
n}\right)  M_{r,\kappa}^{\downarrow}\left(  \mathrm{d}\mathbf{z}_{r,\kappa
}^{\downarrow}\mathrm{d}\mathbf{z}_{r,\kappa+1}^{\uparrow}\right)
\end{align*}%
\vskip-\lastskip
Tending $\theta^{\prime}\rightarrow\theta$, the last end converges to $0$, and
so does the left most side. Hence, due to bounded convergence theorem, we have
$\int_{\mathbf{z}_{r,\kappa}^{\downarrow}\in\Delta}\left[  \lim_{\theta
^{\prime}\rightarrow\theta}g_{\theta,\theta^{\prime}}\left(  \mathbf{z}%
_{r,\kappa}^{\downarrow}\right)  \right]  P_{\theta}^{\mathsf{M}_{r,\kappa
}^{\downarrow}}\left(  \mathrm{d}\mathbf{z}_{r,\kappa}^{\downarrow}\right)
=0$ for any Borel set $\Delta$. Therefore, $g_{\theta,\theta^{\prime}}\left(
\mathbf{z}_{r,\kappa}^{\downarrow}\right)  =0$ for $P_{\theta}$-a.e. Since
$P_{\theta}\left(  \Delta\right)  =\mathrm{tr}\,\rho_{\theta}M\left(
\Delta\right)  =0$ $\Rightarrow\rho_{\theta}M\left(  \Delta\right)  =0$ due to
$\rho_{\theta}>0$, the proof is complete.
\end{proof}

\begin{lemma}
\label{lem:leibniz}Suppose (E') is satisfied. Suppose \ also (a) $\rho
_{\theta}>0$, $\forall\theta\in\Theta$, or (b) the estimator is one-way
semi-classical. Then, we have Leibniz rule:%
\begin{equation}
\left.  \partial_{i}\mathrm{E}_{\theta}\left[  T_{n}\right]  \right\vert
_{\theta=\theta_{0}}=\left[  \partial_{i}\mathrm{E}_{\theta}\mathrm{E}%
_{\theta_{0}}\left[  T_{n}|\,\mathfrak{B}_{r,\kappa}^{\downarrow}\right]
+\partial_{i}\mathrm{E}_{\theta_{0}}\mathrm{E}_{\theta}\left[  T_{n}%
|\,\mathfrak{B}_{r,\kappa}^{\downarrow}\right]  \,\right]  _{\theta=\theta
_{0}}. \label{leibniz}%
\end{equation}

\end{lemma}

%

\vskip-\lastskip
(\ref{leibniz-pre}) is a special case of (\ref{leibniz}). To see this, observe
that the estimator (\ref{def-est}) is viewed as semi-classical considering the
quantum statistical model $\left\{  \rho_{\theta}^{\otimes n_{1}}\right\}
_{\theta\in\Theta}$.

\begin{proof}%
\begin{align}
&  \left.  \partial_{i}\mathrm{E}_{\theta}\left[  T_{n}\right]  \right\vert
_{\theta=\theta_{0}}=\left.  \partial_{i}\mathrm{E}_{\theta}\mathrm{E}%
_{\theta}\left[  T_{n}|\,\mathfrak{B}_{r,\kappa}^{\downarrow}\right]
\right\vert _{\theta=\theta_{0}}\nonumber\\
&  =\lim_{\theta^{\prime}\rightarrow\theta_{0}}\left[  \frac{\mathrm{E}%
_{\theta^{\prime}}\mathrm{E}_{\theta_{0}}\left[  T_{n}|\,\mathfrak{B}%
_{r,\kappa}^{\downarrow}\right]  -\mathrm{E}_{\theta_{0}}\mathrm{E}%
_{\theta_{0}}\left[  T_{n}|\,\mathfrak{B}_{r,\kappa}^{\downarrow}\right]
}{\left\Vert \theta^{\prime}-\theta_{0}\right\Vert }+\frac{\mathrm{E}%
_{\theta^{\prime}}\left[  \,\mathrm{E}_{\theta^{\prime}}\left[  T_{n}%
|\,\mathfrak{B}_{r,\kappa}^{\downarrow}\right]  -\mathrm{E}_{\theta_{0}%
}\left[  T_{n}|\,\mathfrak{B}_{r,\kappa}^{\downarrow}\right]  \,\right]
}{\left\Vert \theta^{\prime}-\theta_{0}\right\Vert }\right]  ,
\label{leibniz-proof}%
\end{align}
where the convention is that $\theta^{\prime j}=\theta_{0}^{j}$ ($j\neq i$).
The first term converges to $\partial_{i}\mathrm{E}_{\theta}\mathrm{E}%
_{\theta_{0}}\left[  T_{n}|\,\mathfrak{B}_{r,\kappa}^{\downarrow}\right]
_{\theta=\theta_{0}}$ , due \ $\left\vert \mathrm{E}_{\theta_{0}}\left[
T_{n}|\,\mathfrak{B}_{r,\kappa}^{\downarrow}\right]  \right\vert <a_{4,n}$ and
Lemma\thinspace\ref{lem:est-cont}. Observe that the second term should
converge due to the convergence of (\ref{leibniz-proof}) and the first term.
Moreover,
\begin{align*}
&  \left\vert \frac{\mathrm{E}_{\theta^{\prime}}\left[  \,\mathrm{E}%
_{\theta^{\prime}}\left[  T_{n}^{j}|\,\mathfrak{B}_{r,\kappa}^{\downarrow
}\right]  -\mathrm{E}_{\theta_{0}}\left[  T_{n}^{j}|\,\mathfrak{B}_{r,\kappa
}^{\downarrow}\right]  \,\right]  }{\left\Vert \theta^{\prime}-\theta
_{0}\right\Vert }-\frac{\mathrm{E}_{\theta_{0}}\left[  \,\mathrm{E}%
_{\theta^{\prime}}\left[  T_{n}^{j}|\,\mathfrak{B}_{r,\kappa}^{\downarrow
}\right]  -\mathrm{E}_{\theta_{0}}\left[  T_{n}^{j}|\,\mathfrak{B}_{r,\kappa
}^{\downarrow}\right]  \,\right]  }{\left\Vert \theta^{\prime}-\theta
_{0}\right\Vert }\right\vert \\
&  =\left\vert \int\left(  \mathrm{E}_{\theta^{\prime}}\left[  T_{n}%
^{j}|\,\mathfrak{B}_{r,\kappa}^{\downarrow}\right]  -\mathrm{E}_{\theta_{0}%
}\left[  T_{n}^{j}|\,\mathfrak{B}_{r,\kappa}^{\downarrow}\right]  \right)
\mathrm{tr}\,\left(  \frac{\rho_{\theta^{\prime}}^{\otimes n}-\rho_{\theta
_{0}}^{\otimes n}}{\left\Vert \theta^{\prime}-\theta_{0}\right\Vert }\right)
M_{\theta_{0},r,\kappa}^{\downarrow}\left(  \mathrm{d}\,\mathbf{z}_{r,\kappa
}^{\downarrow}\right)  \right\vert \\
&  \underset{\text{(i)}}{=}\left\vert \frac{\partial}{\partial\theta^{i}%
}\left[  \int\left(  \mathrm{E}_{\theta^{\prime}}\left[  T_{n}^{j}%
|\,\mathfrak{B}_{r,\kappa}^{\downarrow}\right]  -\mathrm{E}_{\theta_{0}%
}\left[  T_{n}^{j}|\,\mathfrak{B}_{r,\kappa}^{\downarrow}\right]  \right)
\mathrm{tr}\,\rho_{\theta}^{\otimes n}M_{\theta_{0},r,\kappa}^{\downarrow
}\left(  \mathrm{d}\,\mathbf{z}_{r,\kappa}^{\downarrow}\right)  \right]
_{\theta=\tilde{\theta}}\right\vert \\
&  \underset{\text{(ii)}}{=}\left\vert \int\left(  \mathrm{E}_{\theta^{\prime
}}\left[  T_{n}^{j}|\,\mathfrak{B}_{r,\kappa}^{\downarrow}\right]
-\mathrm{E}_{\theta_{0}}\left[  T_{n}^{j}|\,\mathfrak{B}_{r,\kappa
}^{\downarrow}\right]  \right)  \mathrm{tr}\,\partial_{i}\rho_{\tilde{\theta}%
}^{\otimes n}M_{\theta_{0},r,\kappa}^{\downarrow}\left(  \mathrm{d}%
\,\mathbf{z}_{r,\kappa}^{\downarrow}\right)  \right\vert \\
&  \leq\int\left\vert \mathrm{E}_{\theta^{\prime}}\left[  T_{n}^{j}%
|\,\mathfrak{B}_{r,\kappa}^{\downarrow}\right]  -\mathrm{E}_{\theta_{0}%
}\left[  T_{n}^{j}|\,\mathfrak{B}_{r,\kappa}^{\downarrow}\right]  \right\vert
\mathrm{tr}\,\Diamond_{i,\theta_{0},n}^{\left(  1\right)  }M_{\theta
_{0},r,\kappa}^{\downarrow}\left(  \mathrm{d}\,\mathbf{z}_{r,\kappa
}^{\downarrow}\right)  \underset{\text{(iii)}}{\rightarrow}0\quad\left(
\theta^{\prime}\rightarrow\theta_{0}\right)  .
\end{align*}
Here, (i) is due to mean value theorem, where $\tilde{\theta}$ is a point
between $\theta^{\prime}$ and $\theta_{0}$. (ii) is due to \newline$\left\vert
\mathrm{E}_{\theta^{\prime}}\left[  T_{n}^{j}|\,\mathfrak{B}_{r,\kappa
}^{\downarrow}\right]  -\mathrm{E}_{\theta_{0}}\left[  T_{n}^{j}%
|\,\mathfrak{B}_{r,\kappa}^{\downarrow}\right]  \right\vert \leq2a_{4,n}$ and
Lemma\thinspace\ref{lem:est-cont}. \ (iii) is due to Lemma\thinspace
\ref{lem:ec-weak}. Therefore, the second term equals \newline$\underset
{\theta^{\prime}\rightarrow\theta_{0}}{\lim}\frac{\mathrm{E}_{\theta_{0}%
}\left[  \,\mathrm{E}_{\theta^{\prime}}\left[  T_{n}^{j}|\,\mathfrak{B}%
_{r,\kappa}^{\downarrow}\right]  -\mathrm{E}_{\theta_{0}}\left[  T_{n}%
^{j}|\,\mathfrak{B}_{r,\kappa}^{\downarrow}\right]  \,\right]  }{\left\Vert
\theta^{\prime}-\theta_{0}\right\Vert }=\partial_{i}\mathrm{E}_{\theta_{0}%
}\mathrm{E}_{\theta}\left[  T_{n}^{j}|\,\mathfrak{B}_{r,\kappa}^{\downarrow
}\right]  \,_{\theta=\theta_{0}}$. After all, we have (\ref{leibniz}).
\end{proof}

\vskip-4mm

\subsection*{{\protect\normalsize \noindent}\noindent{\protect\normalsize 4.4
ON LOGARITHMIC\ DERIVATIVE}}

\vskip-4mm

Applying Leibniz rule to the indicator function, we can prove that $\left[
\partial_{i}\int_{\Delta}P_{\theta}^{\mathsf{M}_{r,\kappa+1}^{\uparrow
,\mathbf{z}_{r}}}\left(  \Delta^{\prime}\right)  \mathrm{d}P_{\theta_{0}%
}^{\mathsf{M}_{r,\kappa}^{\downarrow}}\left(  z_{r,\kappa}\right)  \,\right]
_{\theta=\theta_{0}}$ is finite. However, in general, one cannot prove
existence of $\partial_{i}P_{\theta}^{\mathsf{M}_{r,\kappa+1}^{\uparrow
,\mathbf{z}_{r}}}\left(  \Delta^{\prime}\right)  $. Therefore, we cannot
define logarithmic derivative of the conditional probability distribution
$P_{\theta}^{\mathsf{M}_{r,\kappa+1}^{\uparrow,\mathbf{z}_{r}}}$, nor cannot
use the argument in Hayashi and Matsumot\.{o}\thinspace(1998) in
semi-classical case.

In one-way semi-classical case, which is treated in Hayashi and Matsumot\.{o}%
\thinspace(1998), one can safely define the logarithmic derivative of
$P_{\theta}^{\mathsf{M}_{r,\kappa+1}^{\uparrow,\mathbf{z}_{r}}}$, since
$P_{\theta}^{\mathsf{M}_{r+1,r+1}^{\uparrow,\mathbf{z}_{r}}}\left(
\mathrm{d}\,z_{r+1,r+1}\right)  =\mathrm{tr}\,\rho_{\theta}M_{r,\kappa
+1}^{\uparrow,\mathbf{z}_{r}}\left(  \mathrm{d}\,z_{r+1,r+1}\right)  $.
Therefore, their argument can be made regorous, though we do not go into detail.

\subsection*{{\protect\normalsize \noindent}\noindent{\protect\normalsize 4.5
PROOF OF THEOREM\thinspace\ref{th:cr-cl}}}

\vskip-3mm Observe that the estimator (\ref{def-est}) with $n_{1}=1$ is
one-way semi-classical. Therefore, the achievability by (one-way)
semi-classical measurement follows from Lemmas\thinspace\ref{lem:est-compose}%
-\ref{lem:est-asym-unbiased}. Therefore, below we prove the lowerbound.

In case , \ $\rho_{\theta}>0$ ($\forall\theta\in\Theta$), due to the proof of
lowerbound part of Theorem\thinspace\ref{th:cr}, we have the following lowerbound.

$\quad\quad C_{\theta}\left(  G_{\theta},\mathcal{M}\right)  \geq
\underset{n\rightarrow\infty}{\varlimsup}\inf\left\{  n\mathrm{Tr}\,G_{\theta
}\mathrm{V}_{\theta}\left[  \mathcal{E}_{\theta,n}\right]  \,\text{\thinspace
};\,\text{semi-classical, (\ref{locally-unbiased}), (E') }\right\}  $\newline
In the following, we reduce the optimization over semi-classical measurements
to the one over \textit{independent semi-classical measurements}, or one-way
semi-classical measurements such that $\mathsf{N}_{\theta_{0},\kappa}$ acting
on $\kappa$-th sample cannot depend on the data $y_{\kappa^{\prime}}$ from
$\mathsf{N}_{\theta_{0},\kappa^{\prime}}$ ($\kappa\neq\kappa^{\prime}$).

\begin{lemma}
\label{lem:semi-classical}Suppose that semi-classical estimator $\mathcal{E}%
_{\theta_{0},n}=\{T_{\theta_{0},n},\mathsf{M}_{\theta_{0}}^{n}\}$ satisfies
(\ref{locally-unbiased}), (E'). Suppose also $\rho_{\theta}>0$, $\forall
\theta\in\Theta$. Then, we can find an estimator $\mathcal{F}_{n,\theta_{0}%
}=\{S_{\theta_{0},n},\mathsf{N}_{\theta_{0}}^{n}\}$ , such that $\mathsf{N}%
_{\theta_{0}}^{n}$ is independent semi-classical, $\mathrm{V}_{\theta_{0}%
}\left[  \mathcal{F}_{\theta_{0},n}\right]  \leq\mathrm{V}_{\theta_{0}}\left[
\mathcal{E}_{\theta_{0},n}\right]  $, (\ref{locally-unbiased}), and (E') hold.
Moreover, $S_{\theta_{0},n}$ is in the form of (\ref{est-sum}), where
$F_{\theta_{0},\kappa}$ is the function of the data $y_{\kappa}$ from
$\mathsf{N}_{\theta_{0},\kappa}$, such that $\mathrm{E}_{\theta_{0}}\left[
F_{\theta_{0},\kappa}\right]  =0$ :
\begin{equation}
S_{\theta_{0},n}=\sum_{\kappa=1}^{n}F_{\theta_{0},\kappa}\left(  y_{\kappa
}\right)  +\theta_{0}\,. \label{est-sum}%
\end{equation}

\end{lemma}

\begin{proof}
Since $\mathrm{E}_{\theta}\left[  T_{\theta,n}|\mathfrak{B}_{r,\kappa
+1}^{\downarrow}\right]  $ satisfies (E'), we apply Leibniz rule
(\ref{leibniz}) recursively to obtain%
\begin{align*}
&  \partial_{i}\mathrm{E}_{\theta}\left[  T_{\theta_{0},n}\right]
_{\theta=\theta_{0}}=\frac{\partial}{\partial\theta^{i}}\left(  \mathrm{E}%
_{\theta_{0}}\left[  \mathrm{E}_{\theta}\left[  T_{\theta_{0},n}%
|\mathfrak{B}_{R_{n},n-1}^{\downarrow}\right]  \right]  \right)
_{\theta=\theta_{0}}+\frac{\partial}{\partial\theta^{i}}\left(  \mathrm{E}%
_{\theta}\left[  \mathrm{E}_{\theta_{0}}\left[  T_{\theta_{0},n}%
|\mathfrak{B}_{R_{n},n-1}^{\downarrow}\right]  \right]  \right)
_{\theta=\theta_{0}}\\
&  =\frac{\partial}{\partial\theta^{i}}\left(  \mathrm{E}_{\theta_{0}}\left[
\mathrm{E}_{\theta}\left[  T_{\theta_{0},n}|\mathfrak{B}_{R_{n},n-1}%
^{\downarrow}\right]  \right]  \right)  _{\theta=\theta_{0}}+\frac{\partial
}{\partial\theta^{i}}\left(  \mathrm{E}_{\theta_{0}}\left[  \mathrm{E}%
_{\theta}\left[  \mathrm{E}_{\theta_{0}}\left[  T_{\theta_{0},n}%
|\mathfrak{B}_{R_{n},n-1}^{\downarrow}\right]  |\mathfrak{B}_{R_{n}%
,n-2}^{\downarrow}\right]  \right]  \right)  _{\theta=\theta_{0}}\\
&  +\frac{\partial}{\partial\theta^{i}}\left(  \mathrm{E}_{\theta}\left[
\mathrm{E}_{\theta_{0}}\left[  \mathrm{E}_{\theta_{0}}\left[  T_{\theta_{0}%
,n}|\mathfrak{B}_{R_{n},n-1}^{\downarrow}\right]  |\mathfrak{B}_{R_{n}%
,n-2}^{\downarrow}\right]  \right]  \right)  _{\theta=\theta_{0}}\\
&  =\sum_{r=1}^{R_{n}}\sum_{\kappa=1}^{n}\frac{\partial}{\partial\theta^{i}%
}\left(  \mathrm{E}_{\theta_{0}}\left[  \mathrm{E}_{\theta_{0}}\left[
\cdots\mathrm{E}_{\theta}\left[  \mathrm{E}_{\theta_{0}}\left[  \cdots
\mathrm{E}_{\theta_{0}}\left[  T_{\theta_{0},n}|\mathfrak{B}_{R_{n}%
,n-1}^{\downarrow}\right]  \cdots|\mathfrak{B}_{r,\kappa+1}^{\downarrow
}\right]  |\mathfrak{B}_{r,\kappa}^{\downarrow}\right]  \cdots|\mathfrak{B}%
_{1,1}^{\downarrow}\right]  \right]  \right)  _{\theta=\theta_{0}}\\
&  =\sum_{r=1}^{R_{n}}\sum_{\kappa=1}^{n}\frac{\partial}{\partial\theta^{i}%
}\left(  \mathrm{E}_{\theta_{0}}\left[  \mathrm{E}_{\theta}\left[
\mathrm{E}_{\theta_{0}}\left[  T_{\theta_{0},n}|\mathfrak{B}_{r,\kappa
+1}^{\downarrow}\right]  |\mathfrak{B}_{r,\kappa}^{\downarrow}\right]
\right]  \right)  _{\theta=\theta_{0}}.
\end{align*}%
\vskip-\lastskip
Observe that, conditioned by $\mathfrak{B}_{r-1}$, $\ $the random variable
$Z_{r,\kappa}$ and $Z_{r,\kappa^{\prime}}$ are independent, due to the
composition of the measurement. Therefore, due to Fubini's theorem,
\begin{align*}
&  \mathrm{E}_{\theta_{0}}\left[  \mathrm{E}_{\theta}\left[  \mathrm{E}%
_{\theta_{0}}\left[  T_{\theta_{0},n}|\mathfrak{B}_{r,\kappa+1}^{\downarrow
}\right]  |\mathfrak{B}_{r,\kappa}^{\downarrow}\right]  |\mathfrak{B}%
_{r-1}\right]  =\int\mathrm{E}_{\theta_{0}}\left[  T_{\theta_{0}%
,n}|\mathfrak{B}_{r}\right]  \prod_{\kappa^{\prime}:\kappa^{\prime}\neq\kappa
}\mathrm{d}\,P_{\theta_{0}}\left(  z_{r,\kappa^{\prime}}|\mathfrak{B}%
_{r-1}\right)  \mathrm{d}\,P_{\theta}\left(  z_{r,\kappa}|\mathfrak{B}%
_{r-1}\right) \\
&  =\mathrm{E}_{\theta}\left[  \left[  \mathrm{E}_{\theta_{0}}\left[
T_{\theta_{0},n}|\left\langle \mathfrak{B}_{r-1},\mathfrak{B}_{r,\kappa
}\right\rangle \right]  |\mathfrak{B}_{r-1}\right]  \right]  .
\end{align*}%
\vskip-\lastskip
Therefore,
\begin{equation}
\partial_{i}\mathrm{E}_{\theta}\left[  T_{\theta_{0},n}\right]  _{\theta
=\theta_{0}}=\sum_{r=1}^{R_{n}}\sum_{\kappa=1}^{n}\frac{\partial}%
{\partial\theta^{i}}\left(  \mathrm{E}_{\theta_{0}}\left[  \mathrm{E}_{\theta
}\left[  \mathrm{E}_{\theta_{0}}\left[  T_{\theta_{0},n}|\left\langle
\mathfrak{B}_{r-1},\mathfrak{B}_{r,\kappa}\right\rangle \right]
|\mathfrak{B}_{r-1}\right]  \right]  \right)  _{\theta=\theta_{0}}.
\label{decom-derivative}%
\end{equation}%
\vskip-\lastskip
Let us define, with the convention $\mathfrak{B}_{0}=\{\emptyset,%
\mathbb{R}
^{l}\}$,

$\quad\quad f_{\theta_{0},r,\kappa}:=\mathrm{E}_{\theta_{0}}\left[
T_{\theta_{0},n}|\,\left\langle \mathfrak{B}_{r-1},\mathfrak{B}_{r,\kappa
}\right\rangle \right]  -\mathrm{E}_{\theta_{0}}\left[  T_{n}|\mathfrak{B}%
_{r-1}\right]  ,\quad F_{\theta_{0},\kappa}:=\sum_{r=1}^{R_{n}}f_{\theta
_{0},r,\kappa}\,.$ $\ $\newline Since $f_{\theta_{0},r,\kappa}$ also satisfies
(E'), we can apply Leibniz rule (\ref{leibniz}). Therefore,
\begin{align}
&  \left.  \partial_{i}\mathrm{E}_{\theta}f_{\theta_{0},r,\kappa}\right\vert
_{\theta=\theta_{0}}=\left.  \partial_{i}\mathrm{E}_{\theta}\left[
\mathrm{E}_{\theta}\left[  f_{\theta_{0},r,\kappa}\,|\mathfrak{B}%
_{r-1}\right]  \right]  \right\vert _{\theta=\theta_{0}}\nonumber\\
&  =\left(  \partial_{i}\mathrm{E}_{\theta_{0}}\left[  \mathrm{E}_{\theta
}\left[  f_{\theta_{0},r,\kappa}\,|\mathfrak{B}_{r-1}\right]  \right]
+\partial_{i}\mathrm{E}_{\theta}\left[  \mathrm{E}_{\theta_{0}}\left[
f_{\theta_{0},r,\kappa}\,|\mathfrak{B}_{r-1}\right]  \right]  \,\,\right)
_{\theta=\theta_{0}}\underset{\text{(i)}}{=}\left.  \partial_{i}%
\mathrm{E}_{\theta_{0}}\left[  \mathrm{E}_{\theta}\left[  f_{\theta
_{0},r,\kappa}\,|\mathfrak{B}_{r-1}\right]  \right]  \right\vert
_{\theta=\theta_{0}}\nonumber\\
&  =\partial_{i}\left(  \mathrm{E}_{\theta_{0}}\left[  \mathrm{E}_{\theta
}\left[  \left(  \mathrm{E}_{\theta_{0}}\left[  T_{\theta_{0},n}%
|\,\left\langle \mathfrak{B}_{r-1},\mathfrak{B}_{r,\kappa}\right\rangle
\right]  -\mathrm{E}_{\theta_{0}}\left[  T_{\theta_{0},n}|\,\mathfrak{B}%
_{r-1}\right]  \right)  |\mathfrak{B}_{r-1}\right]  \,\right]  \right)
_{\theta=\theta_{0}}\nonumber\\
&  \underset{\text{(ii)}}{=}\left.  \partial_{i}\mathrm{E}_{\theta_{0}}\left[
\mathrm{E}_{\theta}\left[  \mathrm{E}_{\theta_{0}}\left[  T_{\theta_{0}%
,n}|\,\left\langle \mathfrak{B}_{r-1},\mathfrak{B}_{r,\kappa}\right\rangle
\right]  |\mathfrak{B}_{r-1}\right]  \right]  \right\vert _{\theta=\theta_{0}%
}. \label{dEf=}%
\end{align}%
\vskip-\lastskip
Here, (i) is due to $\mathrm{E}_{\theta_{0}}\left[  f_{\theta_{0},r,\kappa
}\,|\mathfrak{B}_{r-1}\right]  =\mathrm{E}_{\theta_{0}}\left[  \left(
\mathrm{E}_{\theta_{0}}\left[  T_{\theta_{0},n}|\,\left\langle \mathfrak{B}%
_{r-1},\mathfrak{B}_{r,\kappa}\right\rangle \right]  -\mathrm{E}_{\theta_{0}%
}\left[  T_{n}|\mathfrak{B}_{r-1}\right]  \,\right)  \,|\mathfrak{B}%
_{r-1}\right]  =0$, and (ii) is due to $\mathrm{E}_{\theta}\left[
\mathrm{E}_{\theta_{0}}\left[  T_{\theta_{0},n}|\mathfrak{B}_{r-1}\right]
|\mathfrak{B}_{r-1}\right]  =\mathrm{E}_{\theta_{0}}\left[  T_{\theta_{0}%
,n}|\mathfrak{B}_{r-1}\right]  $. Combining (\ref{decom-derivative}) and
(\ref{dEf=}), we have $\left.  \partial_{i}\mathrm{E}_{\theta}\left[
T_{\theta_{0},n}\right]  \right\vert _{\theta=\theta_{0}}=\sum_{r=1}^{R_{n}%
}\sum_{\kappa=1}^{n}\left.  \partial_{i}\mathrm{E}_{\theta}\left[
f_{\theta_{0},r,\kappa}\right]  \right\vert _{\theta=\theta_{0}}$. Therefore,
with $S_{\theta_{0},n}^{\prime}:=\sum_{\kappa=1}^{n}F_{\theta_{0},\kappa
}\left(  \mathbf{Z}_{R_{n}}\right)  +\theta_{0}$, $\left\{  \mathsf{M}%
_{\theta_{0}}^{n},S_{\theta_{0},n}^{\prime}\right\}  $ is locally unbiased at
$\theta_{0}$. Also, observe the following relations:

$\quad\mathrm{E}_{\theta_{0}}f_{\theta_{0},r,\kappa}\left(  f_{\theta
_{0},r^{\prime},\kappa^{\prime}}\right)  ^{T}=0\quad\left(  \kappa\neq
\kappa^{\prime}\quad\text{or}\quad r\neq r^{\prime}\right)  ,\quad
\mathrm{E}_{\theta_{0}}\,f_{\theta_{0},r,\kappa}\left(  T_{\theta_{0}%
,n}\right)  ^{T}=\mathrm{E}_{\theta_{0}}\,f_{\theta_{0},r,\kappa}\left(
f_{\theta_{0},r,\kappa}\right)  ^{T}.$\newline Due to them, the variance of
this estimate is not larger than the one of $T_{\theta_{0},n}$:\newline%
$\quad\quad\mathrm{V}_{\theta_{0}}\left[  \left\{  \mathsf{M}_{\theta_{0}}%
^{n},S_{\theta_{0},n}^{\prime}\right\}  \right]  =\sum_{\kappa=1}^{n}%
\sum_{r=1}^{R_{n}}\mathrm{E}_{\theta_{0}}f_{\theta_{0},r,\kappa}\left(
f_{\theta_{0},r,\kappa}\right)  ^{T}=\sum_{\kappa=1}^{n}\sum_{r=1}^{R_{n}%
}\mathrm{E}_{\theta_{0}}\left[  f_{\theta_{0},r,\kappa}\left(  T_{\theta
_{0},n}\right)  ^{T}\right]  \leq\mathrm{V}_{\theta_{0}}\left[  \mathcal{E}%
_{\theta_{0},n}\right]  $.

Below, we define $\mathsf{N}_{\theta_{0},\kappa}^{n}$ . First, using
$\rho_{\theta}^{\otimes n}$, we prepare $n$ of \textit{fake ensembles}
$\rho_{\theta_{0}}\otimes\rho_{\theta_{0}}\cdots\otimes\underset{\kappa}%
{\rho_{\theta}}\otimes\cdots\otimes\rho_{\theta_{0}}$ ($\kappa=1$,$\cdots$%
,$n$), composed with single $\rho_{\theta}$ and $n-1$ of $\rho_{\theta_{0}}$.
Then $\mathsf{N}_{\theta_{0},\kappa}^{n}$ is the application of $\mathsf{M}%
_{\theta_{0}}^{n}$ to $\kappa$-th fake ensemble. We denote by $z_{r,\kappa
^{\prime}}^{\left(  \kappa\right)  }$ the data obtained at $r$-th round from
the $\kappa^{\prime}$-th (possibly fake) sample in the $\kappa$-th fake
ensemble. The data $y_{\kappa}$ from $\mathsf{N}_{\theta_{0},\kappa}^{n}$ is
$y_{\kappa}:=\mathbf{z}_{R_{n}}^{\left(  \kappa\right)  }$.

If $\theta=\theta_{0}$, $Y_{\kappa}=\mathbf{Z}_{R_{n}}^{\left(  \kappa\right)
}$ obeys the same probability distribution as $\mathbf{Z}_{R_{n}}$, for any
$\kappa$. Therefore, $\mathrm{V}_{\theta_{0}}\left[  \left\{  \mathsf{N}%
_{\theta_{0}}^{n},F_{\theta_{0},\kappa}\left(  Y_{\kappa}\right)  \right\}
\right]  $ equals $\mathrm{V}_{\theta_{0}}\left[  \left\{  \mathsf{M}%
_{\theta_{0}}^{n},F_{\theta_{0},\kappa}\left(  \mathbf{Z}_{R_{n}}\right)
\right\}  \right]  $. Therefore, due to $\mathrm{V}_{\theta_{0}}\left[
\left\{  \mathsf{M}_{\theta_{0}}^{n},S_{\theta_{0},n}^{\prime}\right\}
\right]  =\sum_{\kappa=1}^{n}\mathrm{V}_{\theta_{0}}\left[  F_{\theta
_{0},\kappa}\right]  $, \ we have $\mathrm{V}_{\theta_{0}}\left[
\mathcal{F}_{n,\theta_{0}}\right]  =\mathrm{V}_{\theta_{0}}\left[  \left\{
\mathsf{M}_{\theta_{0}}^{n},\,S_{\theta_{0},n}^{\prime}\right\}  \right]
\leq\mathrm{V}_{\theta_{0}}\left[  \mathcal{E}_{\theta_{0},n}\right]
$\textrm{. \ \ }Analogously, we can also show \textrm{\ }$\mathrm{E}%
_{\theta_{0}}^{\mathsf{N}_{\theta_{0}}^{n}}\left[  S_{\theta_{0},n}\right]
=\mathrm{E}_{\theta_{0}}\left[  S_{\theta_{0},n}^{\prime}\right]  =\theta_{0}$.

Finally, we show $\left.  \partial_{i}\mathrm{E}_{\theta}^{\mathsf{N}%
_{\theta_{0}}^{n}}\left[  S_{\theta_{0},n}^{j}\right]  \right\vert
_{\theta=\theta_{0}}=\left.  \partial_{i}\mathrm{E}_{\theta}\left[
S_{\theta_{0},n}^{\prime\,j}\right]  \right\vert _{\theta=\theta_{0}}%
=\delta_{i}^{j}$. Observe
\begin{align*}
&  \frac{\partial}{\partial\theta^{i}}\left(  \mathrm{E}_{\theta}%
^{\mathsf{N}_{\theta_{0}}^{n}}f_{\theta_{0},\kappa,r}\left(  \mathbf{Z}%
_{r-1}^{\left(  \kappa\right)  },Z_{r,\kappa}^{\left(  \kappa\right)
}\right)  \right)  _{\theta=\theta_{0}}\\
&  =\frac{\partial}{\partial\theta^{i}}\left[  \int f_{\theta_{0},\kappa
,r}\left(  \mathbf{z}_{r-1}^{\left(  \kappa\right)  },z_{r,\kappa}^{\left(
\kappa\right)  }\right)  \prod_{r^{\prime}=1}^{r}\left(  \prod_{\kappa
^{\prime}:\kappa^{\prime}\neq\kappa}\mathrm{\,d}P_{\theta_{0}}\left(
z_{r^{\prime},\kappa^{\prime}}^{\left(  \kappa\right)  }|\mathbf{z}%
_{r^{\prime}-1}^{\left(  \kappa\right)  }\right)  \mathrm{d}P_{\theta}\left(
z_{r^{\prime},\kappa}^{\left(  \kappa\right)  }|\mathbf{z}_{r^{\prime}%
-1}^{\left(  \kappa\right)  }\right)  \right)  \right]  _{\theta=\theta_{0}}\\
&  =\partial_{i}\left(  \mathrm{E}_{\theta}\left[  \mathrm{E}_{\theta_{0}%
}\left[  \cdots\mathrm{E}_{\theta}\left[  \mathrm{E}_{\theta_{0}}\left[
\mathrm{E}_{\theta}\left[  f_{\theta_{0},\kappa,r}|\mathfrak{B}_{r-1}\right]
|\left\langle \mathfrak{B}_{r-1,\kappa},\mathfrak{B}_{r-2}\right\rangle
\right]  \,|\,\mathfrak{B}_{r-2}\right]  \cdots|\mathfrak{B}_{1,\kappa
}\right]  \right]  \right)  _{\theta=\theta_{0}}\\
&  =\partial_{i}\left(  \mathrm{E}_{\theta_{0}}\left[  \mathrm{E}_{\theta
}\left[  f_{\theta_{0},\kappa,r}|\mathfrak{B}_{r-1}\right]  \right]  \right)
_{\theta=\theta_{0}}+\sum_{r^{\prime}=3}^{r-1}\partial_{i}\left(
\mathrm{E}_{\theta_{0}}\left[  \mathrm{E}_{\theta}\left[  \mathrm{E}%
_{\theta_{0}}\left[  f_{\theta_{0},\kappa,r}|\left\langle \mathfrak{B}%
_{r^{\prime},\kappa},\mathfrak{B}_{r^{\prime}-1}\right\rangle \right]
\,|\,\mathfrak{B}_{r^{\prime}-2}\right]  \right]  \right)  _{\theta=\theta
_{0}}\\
&  +\partial_{i}\left(  \mathrm{E}_{\theta}\left[  \mathrm{E}_{\theta_{0}%
}\left[  f_{\theta_{0},\kappa,r}|\mathfrak{B}_{1,\kappa},\right]  \,\right]
\right)  _{\theta=\theta_{0}}%
\end{align*}%
\vskip-\lastskip
where the last equality is due to Leibniz rule. Due to the definition of
$f_{\theta_{0},\kappa,r}$, $\mathrm{E}_{\theta_{0}}\left[  f_{\theta
_{0},\kappa,r}|\left\langle \mathfrak{B}_{r^{\prime},\kappa},\mathfrak{B}%
_{r^{\prime}-1}\right\rangle \right]  $ ($r^{\prime}\leq r-1$) and
$\mathrm{E}_{\theta_{0}}\left[  f_{\theta_{0},\kappa,r}|\mathfrak{B}%
_{1,\kappa},\right]  $ are zero. Therefore, $\left.  \partial_{i}%
\mathrm{E}_{\theta}^{\mathsf{N}_{\theta_{0}}^{n}}\left[  S_{\theta_{0},n}%
^{j}\right]  \right\vert _{\theta=\theta_{0}}=\left.  \partial_{i}%
\mathrm{E}_{\theta}\left[  S_{\theta_{0},n}^{\prime\,j}\right]  \right\vert
_{\theta=\theta_{0}}=\delta_{i}^{j}$ follows from (\ref{dEf=}). \ Trivially,
$\mathcal{F}_{n,\theta_{0}}=\{\mathsf{N}_{\theta_{0}}^{n},\,S_{\theta_{0}%
,n}\}$ satisfies (E'). After all, we have the lemma.
\end{proof}

\begin{lemma}
\label{lem:semi-classical-2}Suppose $\mathsf{N}_{\theta_{0}}^{n}$ is
independent semi-classical. Suppose also that $S_{\theta_{0},n}$ is in the
form of (\ref{est-sum}), and that $\mathcal{F}_{n,\theta_{0}}=\{\mathsf{N}%
_{\theta_{0}}^{n},\,S_{\theta_{0},n}\}$ satisfies (\ref{locally-unbiased}) and
(E'). Then, we can find an estimator $\mathcal{E}_{\theta_{0},1}^{\prime
}=\{\mathsf{M}_{\theta_{0}}^{^{\prime}},T_{\theta_{0},1}^{\prime}\}$ \ acting
on single sample, with (\ref{locally-unbiased}), (E'), and $\mathrm{V}%
_{\theta_{0}}\left[  \mathcal{E}_{\theta_{0},1}^{\prime}\right]
=n\mathrm{V}_{\theta_{0}}\left[  \mathcal{F}_{n,\theta_{0}}\right]  $.
\end{lemma}

\begin{proof}
$\mathsf{M}_{\theta_{0}}^{^{\prime}}$ is constructed as follows; generate
$x_{\kappa}\in\left\{  1,\cdots,n\right\}  $ according to uniform
distribution, and apply $\mathsf{N}_{\theta_{0},x_{\kappa}}$ to $\rho_{\theta
}$, generating the data $y_{\kappa}$.\ The data resulting from $\mathsf{M}%
_{\theta_{0}}^{^{\prime}}$ is the pair $y_{\kappa}^{\prime}:=(x_{\kappa
},y_{x_{\kappa}})$. \ $T_{\theta_{0},1}^{\prime}$ is defined by $T_{\theta
_{0},1}^{\prime}\left(  y_{\kappa}^{\prime}\right)  :=nF_{\theta_{0}%
,x_{\kappa}}\left(  y_{x_{\kappa}}\right)  +\theta_{0}$.

Observe $\mathrm{E}_{\theta}^{\mathsf{M}_{\theta_{0}}^{^{\prime}}}\left[
T_{\theta_{0},1}^{\prime}\right]  =\frac{1}{n}\sum_{\kappa=1}^{n}%
n\mathrm{E}_{\theta}^{\mathsf{N}_{\theta_{0},\kappa}}\left[  F_{\theta
_{0},\kappa}\right]  +\theta_{0}=\mathrm{E}_{\theta}^{\mathsf{N}_{\theta_{0}%
}^{n}}\left[  S_{\theta_{0},n}\right]  $, implying (\ref{locally-unbiased})
for $\mathcal{E}_{\theta_{0}}^{\prime}$. MSE of $\mathcal{E}_{\theta_{0}%
}^{\prime}$ is computed as follows: $\mathrm{V}_{\theta_{0}}\left[
\mathcal{E}_{\theta_{0},1}^{\prime}\right]  =\frac{1}{n}\sum_{\kappa=1}%
^{n}n^{2}\mathrm{E}_{\theta}^{\mathsf{N}_{\theta_{0}}^{n}}\left[
F_{\theta_{0},\kappa}\left(  Y_{\kappa}\right)  F_{\theta_{0},\kappa}\left(
Y_{\kappa}\right)  ^{T}\right]  =n\mathrm{V}_{\theta_{0}}\left[
\mathcal{F}_{n,\theta_{0}}\right]  $. (\ref{locally-unbiased}) and (E') for
$\mathcal{E}_{\theta_{0}}^{\prime}$ are trivial.
\end{proof}

Due to Lemma\thinspace\ref{lem:ec-weak}, the above two lemmas leads to `$\geq
$'- part of the first identity of Theorem\thinspace\ref{th:cr-cl}, in the case
where $\rho_{\theta}>0$, $\forall\theta\in\Theta$. The statement for the
general case is straightforward consequence of the following lemma.

\begin{lemma}
\label{lem:approx-variance}Let $\mathcal{E}_{\theta,n}:=\left\{
\mathsf{M}_{\theta}^{n},\,T_{\theta,n}\right\}  $ be a locally unbiased
estimator at $\theta$ with (E'). Then, $\exists V\geq0$, $\forall
\varepsilon>0$, $\exists\mathcal{E}_{\varepsilon,\theta,1}^{\prime
}:=\{\mathsf{M}_{\theta}^{\prime},T_{\varepsilon,\theta,1}^{\prime}\}$\ acting
on a single sample with (E') and $\mathrm{V}_{\theta}\left[  \mathcal{E}%
_{\theta,n}\right]  \geq\frac{1}{n}\mathrm{V}_{\theta}\left[  \mathcal{E}%
_{\theta,1}^{\prime}\right]  -\frac{\varepsilon}{\left(  1-\varepsilon\right)
}V$.
\end{lemma}

\begin{proof}
Let $\sigma>0$, and define $\rho_{\theta,\varepsilon}:=\left(  1-\varepsilon
\right)  \rho_{\theta}+\varepsilon\sigma$. Denote by $\mathrm{E}%
_{\theta,\varepsilon}\left[  T_{\theta,n}\right]  $ and $\mathrm{V}%
_{\theta,\varepsilon}\left[  \mathcal{E}_{\theta,n}\right]  $ the average and
the variance of $\mathcal{E}_{\theta,n}$ with respect to $\rho_{\theta
,\varepsilon}$. Then, there is a locally unbiased estimator $\mathcal{E}%
_{\theta,\varepsilon,n}:=\left\{  \mathsf{M}_{\theta}^{n},\,T_{\theta
,\varepsilon,n}\right\}  $ with respect to the quantum statistical model
$\left\{  \rho_{\theta,\varepsilon}\right\}  _{\theta\in\Theta}$ which
satisfies (E) and the relation $T_{\theta,\varepsilon,n}=\frac{1}%
{1-\varepsilon}\left(  T_{\theta,n}-\mathrm{E}_{\theta,\varepsilon}\left[
T_{\theta,n}\right]  \right)  +\theta$. Since $\rho_{\theta,\varepsilon}>0$
and the family $\left\{  \rho_{\theta,\varepsilon}\right\}  _{\theta\in\Theta
}$ satisfies (M.1,2), Lemmas\thinspace\ref{lem:semi-classical}%
-\ref{lem:semi-classical-2} imply existence of a locally unbiased estimator
$\mathcal{E}_{\varepsilon,\theta,1}^{\prime}:=\{\mathsf{M}_{\theta}^{\prime
},T_{\varepsilon,\theta,1}^{\prime}\}$ acting on single sample such that (E')
and $\mathrm{V}_{\theta,\varepsilon}\left[  \mathcal{E}_{\theta,\varepsilon
,n}\right]  \geq\frac{1}{n}\mathrm{V}_{\theta,\varepsilon}\left[
\mathcal{E}_{\theta,\varepsilon,1}^{\prime}\right]  $.

Defining $\mathcal{E}_{\theta,1}^{\prime}:=$ $\{\mathsf{M}_{\theta}^{\prime
},T_{\theta,1}^{\prime}\}$ by $T_{\theta,1}^{\prime}:=\left(  1-\varepsilon
\right)  \left(  T_{\theta,\varepsilon,1}^{\prime}-\mathrm{E}_{\theta}\left[
T_{\theta,1}^{\prime}\right]  \right)  +\theta$, $\mathcal{E}_{\theta
,1}^{\prime}$ is locally unbiased with respect to $\left\{  \rho_{\theta
}\right\}  _{\theta\in\Theta}$. \ Letting $V$ be the variance of
$\mathcal{E}_{\theta,n}$ with respect to $\sigma$, we have \newline$\quad
\quad\quad\mathrm{V}_{\theta}\left[  \mathcal{E}_{\theta,n}\right]  =\left(
1-\varepsilon\right)  \mathrm{V}_{\theta,\varepsilon}\left[  \mathcal{E}%
_{\theta,\varepsilon,n}\right]  -\frac{\varepsilon}{1-\varepsilon}V\geq
\frac{1-\varepsilon}{n}\mathrm{V}_{\theta,\varepsilon}\left[  \mathcal{E}%
_{\theta,\varepsilon,1}^{\prime}\right]  -\frac{\varepsilon}{1-\varepsilon
}V\geq\frac{1}{n}\mathrm{V}_{\theta}\left[  \mathcal{E}_{\theta,1}^{\prime
}\right]  -\frac{\varepsilon}{1-\varepsilon}V$.
\end{proof}

Combining with the achievability, the first identity of Theorem\thinspace
\thinspace\ref{th:cr-cl} is proved. The second identity is shown using the
analogous argument as the proof of (\ref{cr-q1})=(\ref{cr-q2}), and
Theorem\thinspace\ref{th:cr-cl} is proved.

Note that the estimator achieving the lowerbound is one-way semi-classical.
This means that the optimal asymptotic cost of one-way semi-classical
measurements is also $C_{\theta}\left(  G_{\theta},\mathcal{M}\right)  $.

\vskip-10mm

\section*{{\protect\normalsize 5. LOCC STATE ESTIMATION}}

\vskip-3mm Suppose the state $\rho_{\theta}^{\otimes n}$ is shared by remote
party, Alice and Bob: $\mathcal{H=H}_{\mathfrak{a}}\otimes\mathcal{H}%
_{\mathfrak{b}}$. Suppose also that Alice and Bob can exchange classical
messages, but cannot interact quantumly with each other. So they do the
following $R_{n}$-round measurement $\mathsf{M}^{n}$: At each round, Alice and
Bob measures her/his share of the samples, and the measurements of the $r$-th
round depend on the previously obtained data. We denote by $\xi_{r}%
^{\mathfrak{a}}$($\in$ $%
\mathbb{R}
^{l}$) and $\xi_{r}^{\mathfrak{b}}$($\in$ $%
\mathbb{R}
^{l}$) the data obtained at the $r$-th round by Alice's and Bob's measurement,
respectively. Also, $\vec{\xi}_{r}$ denotes $\left(  \xi_{r^{\prime}%
}^{\mathfrak{a}},\xi_{r^{\prime}}^{\mathfrak{b}}\right)  _{r^{\prime}=1}^{r}$.
We denote by $\mathfrak{C}_{r}$ and $\mathfrak{C}_{r}^{x}$ ($x=\mathfrak{a}%
,\mathfrak{b}$) the Borel field over the space which $\vec{\xi}_{r-1}$ and
$\xi_{r}^{x}$ takes values in, respectively. The measurement acting in the
$r$-th round by Alice (Bob) is denoted by $\mathsf{M}_{r,\mathfrak{a}}%
^{\vec{\xi}_{r-1}}$ ( $\mathsf{M}_{r,\mathfrak{b}}^{\vec{\xi}_{r-1}}$, resp.).
Such operations are said to be \textit{local operations and quantum
communications }(\textit{LOCC}, in short). An important subclass of LOCC is
\textit{local operation (LO)}, where Alice and Bob does $\mathsf{M}%
_{\mathfrak{a}}^{n}$ and $\mathsf{M}_{\mathfrak{b}}^{n}$, independently.

The difference between semi-classical and LOCC measurements is the split of
the actions. In the former, split is between samples. In the latter, the split
is between Alice and Bob. Other than this point, basically they are the same
concept. Especially, LO corresponds to independent semi-classical measurements.

Define $C_{\theta}^{Q,L}\left(  G_{\theta},\mathcal{M}\right)  $ and
$C_{\theta}^{L}\left(  G_{\theta},\mathcal{M}\right)  $ by restricting the
range of measurement to LOCC in the definition of $C_{\theta}^{Q}\left(
G_{\theta},\mathcal{M}\right)  $ and $C_{\theta}\left(  G_{\theta}%
,\mathcal{M}\right)  $, respectively. Then, trivially, we have:

\begin{theorem}
Suppose (M.1-3) hold. Then, $\quad$\newline$\quad
\ \ \ \ \ \ \ \ \ \ \ \ \ \ \ \ \ \ \ C_{\theta}^{L}\left(  G_{\theta
},\mathcal{M}\right)  =\inf\left\{  \mathrm{Tr}\,G_{\theta_{0}}\mathrm{V}%
_{\theta_{0}}\left[  \mathcal{E}_{\theta_{0},1}\right]  \,\text{\thinspace
};\,\text{ LOCC, (\ref{locally-unbiased}), (E) for }n=1\right\}  $,
\newline$\quad\quad\quad\quad\quad\quad\quad\ C_{\theta}^{Q,L}\left(
G_{\theta},\mathcal{M}\right)  =\underset{n\rightarrow\infty}{\lim}%
\inf\left\{  n\mathrm{Tr}\,G_{\theta_{0}}\mathrm{V}_{\theta_{0}}\left[
\mathcal{E}_{\theta_{0},n}\right]  \,\text{\thinspace};\,\text{ LOCC,
(\ref{locally-unbiased}), (E) }\right\}  $.
\end{theorem}

From here, we focus on the case of $\rho_{\theta}=\rho_{\theta}^{\mathfrak{a}%
}\otimes\rho_{\theta}^{\mathfrak{b}}$. \ The motivation of studying this
seemingly easy case is as follows. Suppose that $\mathrm{rank}\rho_{\theta}%
=1$, $\dim\mathcal{H}<\infty$, and $\rho_{\theta}\neq\rho_{\theta
}^{\mathfrak{a}}\otimes\rho_{\theta}^{\mathfrak{b}}$. Then it is known
that\ $C_{\theta}^{L}=C_{\theta}^{Q,L}=C_{\theta}=C_{\theta}^{Q}$
(Matsumoto\thinspace(2007)). \ The estimator used to show the identity,
however, fails in case of $\rho_{\theta}=\rho_{\theta}^{\mathfrak{a}}%
\otimes\rho_{\theta}^{\mathfrak{b}}$.

In this case, we can translate the argument in Section\thinspace4 about the
lowerbound to the asymptotic cost of semi-classical estimators to LOCC
estimators. To see this, observe the proof of Lemmas\thinspace
\ref{lem:leibniz}-\thinspace\ref{lem:semi-classical} are valid even if
non-identical independent samples $\bigotimes_{\kappa=1}^{n}\rho_{\theta
}^{\left(  \kappa\right)  }$ are given. Below, we present the analogue of
Lemma\thinspace\ref{lem:semi-classical}. (The proof is omitted being almost
parallel.) Using this lemma, the optimization over LOCC is reduced to the one
over LO, where $\mathsf{N}_{\theta_{0},n}^{\mathfrak{a}}$ and $\mathsf{N}%
_{\theta_{0},n}^{\mathfrak{b}}$ is are measured independently, producing the
data $y^{\mathfrak{a}}$,$y^{\mathfrak{b}}$.

\begin{lemma}
Suppose $\rho_{\theta}=\rho_{\theta}^{\mathfrak{a}}\otimes\rho_{\theta
}^{\mathfrak{b}}>0$, $\forall\theta\in\Theta$. Suppose an LOCC estimator
$\mathcal{E}_{\theta_{0},n}=\{\mathsf{M}_{\theta_{0}}^{n},\,T_{\theta_{0}%
,n}\}$ satisfies (\ref{locally-unbiased}) and (E'). Then, we can find an LO
estimator $\mathcal{F}_{\theta_{0},n}=\{\mathsf{N}_{\theta_{0},n}%
,\,S_{\theta_{0},n}\}$ , such that $\mathrm{V}_{\theta_{0}}\left[
\mathcal{F}_{\theta_{0},n}\right]  \leq\mathrm{V}_{\theta_{0}}\left[
\mathcal{E}_{\theta_{0},n}\right]  $, (\ref{locally-unbiased}), and (E') hold.
Moreover, $S_{\theta_{0},n}$ is in the following form:%
\begin{equation}
S_{\theta_{0},n}=F_{\theta_{0},n}^{\mathfrak{a}}\left(  \xi_{n}^{\mathfrak{a}%
}\right)  +F_{\theta_{0},n}^{\mathfrak{b}}\left(  \xi_{n}^{\mathfrak{b}%
}\right)  +\theta_{0}\,. \label{S=F+F}%
\end{equation}

\end{lemma}

%

\vskip-\lastskip
Let $\mathcal{M}_{x}$ $\ $denote $\left\{  \rho_{\theta}^{x}\right\}
_{\theta\in\Theta}$ \ ($x=\mathfrak{a}$,$\mathfrak{b}$). Observe that the map
$\theta\rightarrow\rho_{\theta}^{x}$ may not be injective. Therefore, the
vector space $\left\{  \mathbf{v}_{\theta}^{x};\sum_{i=1}^{m}v^{x,i}%
\partial_{i}\rho_{\theta}^{x}=0\right\}  $ may not be $\{0\}$. We denote by
$\Pi_{\theta}^{x}$ the projector onto the orthogonal complement of this vector
space in $%
\mathbb{R}
^{m}$. \ Letting $\mathcal{F}_{\theta_{0},n}^{x}:=\left\{  \mathsf{N}%
_{\theta_{0},n}^{x},\,F_{\theta_{0},n}^{x}\right\}  $, $\ B_{\theta_{0}%
}\left[  \mathcal{F}_{\theta_{0},n}^{x}\right]  \mathbf{v}_{\theta_{0}}^{x}=0$
if $\Pi_{\theta_{0}}^{x}\mathbf{v}_{\theta_{0}}^{x}=0$. Therefore, there is a
matrix $W_{\theta_{0}}^{x}$ with $B_{\theta_{0}}\left[  \mathcal{F}%
_{\theta_{0},n}^{x}\right]  =W_{\theta_{0}}^{x}\Pi_{\theta_{0}}^{x}$.

We want to minimize the variance of locally unbiased estimator in the form of
(\ref{S=F+F}). First, \ for a given $\left(  F_{\theta_{0},n}^{\mathfrak{a}%
},\,F_{\theta_{0},n}^{\mathfrak{b}}\right)  $, we define $S_{\theta_{0}%
,n}\left[  A_{\theta_{0},n}^{\mathfrak{a}},A_{\theta_{0},n}^{\mathfrak{b}%
}\right]  :=A_{\theta_{0},n}^{\mathfrak{a}}F_{\theta_{0},n}^{\mathfrak{a}%
}\left(  \xi_{n}^{\mathfrak{a}}\right)  +A_{\theta_{0},n}^{\mathfrak{b}%
}F_{\theta_{0},n}^{\mathfrak{b}}\left(  \xi_{n}^{\mathfrak{b}}\right)
+\theta_{0}$, where $\left(  A_{\theta_{0},n}^{\mathfrak{a}},\,A_{\theta
_{0},n}^{\mathfrak{b}}\right)  $\ moves over all the $m\times m$ real
invertible matrices. Elementary but tedious calculation shows that the
variance of such estimators is larger (in the sense that the difference is
positive semi-definite) than
\begin{equation}
\left(  \Pi_{\theta_{0}}^{\mathfrak{a}}\left(  \mathrm{V}_{\theta_{0}}\left[
\mathcal{E}_{\theta_{0},n}^{\mathfrak{a}}\right]  \right)  ^{-1}\Pi
_{\theta_{0}}^{\mathfrak{a}}+\Pi_{\theta_{0}}^{\mathfrak{b}}\left(
\mathrm{V}_{\theta_{0}}\left[  \mathcal{E}_{\theta_{0},n}^{\mathfrak{b}%
}\right]  \right)  ^{-1}\Pi_{\theta_{0}}^{\mathfrak{b}}\right)  ^{-1},
\label{(v-1+v-1)-1}%
\end{equation}%
\vskip-\lastskip
where $\mathcal{E}_{\theta_{0},n}^{x}:=\left\{  \mathsf{N}_{\theta_{0},n}%
^{x},\tilde{F}_{\theta_{0},n}^{x}\right\}  $, $\tilde{F}_{\theta_{0},n}%
^{x}:=\left(  W_{\theta_{0}}^{x}\right)  ^{-1}F_{\theta_{0},n}^{x}$ and
$(\cdot)^{-1}$ in (\ref{(v-1+v-1)-1}) denotes generalized inverse. Observe
that $\mathcal{E}_{\theta_{0},n}^{x}$ ($x=\mathfrak{a}$,$\mathfrak{b}$)
satisfies
\begin{equation}
B_{\theta_{0}}\left[  \mathcal{E}_{\theta_{0},n}^{x}\right]  =\Pi_{\theta_{0}%
}^{x}\quad(x=\mathfrak{a},\mathfrak{b}). \label{locally-unbiased-prime}%
\end{equation}
\
\vskip-\lastskip
In the end, we optimize (\ref{(v-1+v-1)-1}) with the constrain
(\ref{locally-unbiased}) and (E'). \ Here, constrain (E') can be replaced by
(E) without increasing the infimum, due to the analogous argument as the proof
of (\ref{cr-q1})=(\ref{cr-q2}).

So far we had assumed $\rho_{\theta}>0$. However, this assumption can be
removed due to an analogue of Lemma\thinspace\ref{lem:approx-variance} (the
proof is omitted, being straightforward.). Therefore, letting $\mathcal{V}%
_{\theta,n}^{x}\left(  \mathcal{M}\right)  $ the totality of $\mathrm{V}%
_{\theta}\left[  \mathcal{E}_{\theta,n}\right]  $ with
(\ref{locally-unbiased-prime}) and (E) ( $x=\mathfrak{a},\mathfrak{b}$), \ we have

\begin{theorem}
Suppose $\rho_{\theta}=\rho_{\theta}^{\mathfrak{a}}\otimes\rho_{\theta
}^{\mathfrak{b}}$, $\forall\theta\in\Theta$. Then,%
\begin{align*}
C_{\theta}^{L}\left(  G_{\theta},\mathcal{M}\right)   &  \geq\inf\left\{
\mathrm{Tr}\,G_{\theta}\left(  \Pi_{\theta}^{\mathfrak{a}}V_{\mathfrak{a}%
}^{-1}\Pi_{\theta}^{\mathfrak{a}}+\Pi_{\theta}^{\mathfrak{b}}V_{\mathfrak{b}%
}^{-1}\Pi_{\theta}^{\mathfrak{b}}\right)  ^{-1}\text{, }V_{x}\in
\mathcal{V}_{\theta,1}^{x}\left(  \mathcal{M}\right)  ,\,x=\mathfrak{a}%
,\mathfrak{b}\right\}  ,\\
C_{\theta}^{Q,L}\left(  G_{\theta},\mathcal{M}\right)   &  \geq\lim
_{n\rightarrow\infty}\inf\left\{  n\mathrm{Tr}\,G_{\theta}\left(  \Pi_{\theta
}^{\mathfrak{a}}V_{\mathfrak{a}}^{-1}\Pi_{\theta}^{\mathfrak{a}}+\Pi_{\theta
}^{\mathfrak{b}}V_{\mathfrak{b}}^{-1}\Pi_{\theta}^{\mathfrak{b}}\right)
^{-1}\text{, }V_{x}\in\mathcal{V}_{\theta,n}^{x}\left(  \mathcal{M}\right)
,\,x=\mathfrak{a},\mathfrak{b}\right\}  .
\end{align*}

\end{theorem}

\vskip-3mm The achievability of the lowerbound is also true. This theorem
leads to a necessary and sufficient condition for $C_{\theta}^{L}\left(
G_{\theta},\mathcal{M}\right)  $ ($C_{\theta}\left(  G_{\theta},\mathcal{M}%
\right)  $ ) to equal $C_{\theta}^{Q,L}\left(  G_{\theta},\mathcal{M}\right)
$ ($C_{\theta}^{Q}\left(  G_{\theta},\mathcal{M}\right)  $, resp.). These
topics, however, will be discussed elsewhere.

\vskip-10mm

\section*{{\protect\normalsize 6. ESTIMATION OF QUANTUM OPERATIONS}}

\vskip-3mm Suppose we are given a family of completely positive and trace
preserving maps $\mathcal{L}:=\left\{  \Lambda_{\theta}\right\}  _{\theta
\in\Theta}$. Here, $\Lambda_{\theta}:\mathcal{B}\left(  \mathcal{H}\right)
\rightarrow\mathcal{B}\left(  \mathcal{H}^{\prime}\right)  $, $\theta\in
\Theta$, and $\Theta$ is an open region in $%
\mathbb{R}
^{m}$. Our purpose is to estimate $\theta$, by measuring the output of
$\Lambda_{\theta}$ after sending the input state for $n$ times through it.

Our input state $\rho^{n}$ is living in $\mathcal{H}\otimes\mathcal{K}$, where
$\dim\mathcal{K}$ is arbitrarily large. $\mathcal{K}$ may be used to store the
input state before and/or after application of $\Lambda_{\theta}$ . Between
the $\kappa$-th and $\left(  \kappa+1\right)  $-th use of $\Lambda_{\theta}$,
one can apply an operation $\Xi_{\kappa}^{n}:\mathcal{B}\left(  \mathcal{H}%
^{\prime}\otimes\mathcal{K}\right)  \rightarrow\mathcal{B}\left(
\mathcal{H}\otimes\mathcal{K}\right)  $. $\Xi_{\kappa}^{n}$ may be a
composition of measurement followed by preparation of the state to be send
through $\Lambda_{\theta}$. After $n$ times of use of $\Lambda_{\theta}$, we
obtain $\prod_{\kappa=1}^{n}\left\{  \,\left(  \Lambda_{\theta}\otimes
\mathbf{I}\right)  \circ\Xi_{\kappa}^{n}\,\right\}  \left(  \rho^{n}\right)
$. We measure this by $\mathsf{M}^{n}$, obtaining the data $\omega_{n}\in%
\mathbb{R}
^{l_{n}}$, and compute the estimate $T_{n}\left(  \omega_{n}\right)  $. The
pair $\mathcal{E}_{n}:=\{\rho^{n},\left\{  \Xi_{\kappa}^{n}\right\}
_{\kappa=1}^{n-1},\mathsf{M}^{n},T_{n}\}$ (, or sometimes sequence $\left\{
\mathcal{E}_{n}\right\}  _{n=1}^{\infty}$ also, ) is called an
\textit{estimator}. The probability distribution of the data is $P_{\theta
}^{\mathsf{E}_{n}}\left\{  \omega_{n}\in\Delta\right\}  =\mathrm{tr}%
\,M^{n}\left(  \Delta\right)  \prod_{\kappa=0}^{n}\left\{  \left(
\Lambda_{\theta}\otimes\mathbf{I}\right)  \circ\Xi_{\kappa}^{n}\right\}
\left(  \rho^{n}\right)  $.

Regularity conditions, other than (\ref{asym-unbiased}) on estimators, are
listed in Table\thinspace\ref{table:reg-cond-op-model}. Note that they are
honest analogue of (M.1) and (E). In the table, convergence is always in terms
of $\left\Vert \cdot\right\Vert _{cb}$, and $\square_{i,\theta}:\mathcal{B}%
\left(  \mathcal{H}\right)  \rightarrow\mathcal{B}\left(  \mathcal{H}\right)
$ is an affine map with $\square_{i,\theta}\otimes\mathbf{I}\left(
\rho\right)  \geq\partial_{i}\Lambda_{\theta_{0}}\otimes\mathbf{I}\left(
\rho\right)  $ ($\left\Vert \theta_{0}-\theta\right\Vert <a_{2}$), whose
existence is certified by (CM.1) and an analogue of Lemma\thinspace
\ref{lem:diamond}. Also, $\varpi_{1}^{n}\left(  \square_{i,\theta
}\,,\mathsf{E}_{n}\right)  $ is defined by replacing $\partial_{i}%
\Lambda_{\theta}$ in $\partial_{i}\prod_{\kappa=0}^{n}\left\{  \,\left(
\Lambda_{\theta}\otimes\mathbf{I}\right)  \circ\Xi_{\kappa}^{n}\,\right\}
\left(  \rho^{n}\right)  $ by $\square_{i,\theta}$.%

\begin{table}[tbp] \centering
\begin{tabular}
[c]{|l|}\hline
(CM.1) $\partial_{i}\Lambda_{\theta},$ $\partial_{i}\partial_{j}%
\Lambda_{\theta}$ exits and are locally uniformly continuous, $\ \left\Vert
\partial_{i}\Lambda_{\theta}\right\Vert _{cb},$ $\left\Vert \partial
_{i}\partial_{j}\Lambda_{\theta}\right\Vert _{cb}$ $\leq\frac{a_{1}}{2}%
<\infty,\,\forall\theta\in\Theta$\\\hline\hline
(CE) $\ \ \ \int\left\Vert T_{n}\left(  \omega\right)  -\theta\right\Vert
\mathrm{tr}\,\varpi_{1}^{n}\left(  \square_{i,\theta}\,,\mathsf{E}_{n}\right)
M^{n}\left(  \mathrm{d}\omega_{n}\right)  \leq na_{1}a_{4,n}$, $\int\left\Vert
T_{n}\left(  \omega\right)  -\theta\right\Vert ^{2}\mathrm{tr}\,P_{\theta
}^{\mathsf{E}_{n}}\left(  \mathrm{d}\omega_{n}\right)  \rho_{\theta}^{\otimes
n}\leq n\left(  a_{4,n}\right)  ^{2}$, $\,\forall\theta\in\Theta$\\\hline
\end{tabular}
\caption{Regularity conditions on models and estimators in operation
estimation.}\label{table:reg-cond-op-model}%
\end{table}%
\ 

\ Define $C_{\theta}^{Q,Op}\left(  G_{\theta},\mathcal{L}\right)  $ by
replacing (E) in the definition of $C_{\theta}^{Q,Op}\left(  G_{\theta
},\mathcal{M}\right)  $ by (CE). Then we have, honestly modifying the argument
in Section 3 (the proof is omitted), \ 

\begin{theorem}
If (CM.1) holds, $C_{\theta}^{Q,Op}\left(  G_{\theta},\mathcal{L}\right)
\geq\underset{n\rightarrow\infty}{\lim}\inf\left\{  n\mathrm{Tr}\,G_{\theta
}\mathrm{V}_{\theta}\left[  \mathcal{E}_{\theta,n}\right]  \,\text{\thinspace
};\text{ (\ref{locally-unbiased}), (E')}\right\}  $
\end{theorem}

Also, the achievability of the lowerbound can be proved, with some additional
regularity conditions. It is known that in case of $\Lambda_{\theta}\left(
\rho\right)  =U_{\theta}\rho U_{\theta}^{\dagger}$, with $U_{\theta}U_{\theta
}^{\dagger}=U_{\theta}^{\dagger}U_{\theta}=\mathbf{1}$, there is an
asymptotically unbiased estimator with $\mathrm{Tr}\,G_{\theta}\mathrm{V}%
_{\theta}\left[  \mathcal{E}_{n}\right]  =O\left(  \frac{1}{n^{2}}\right)
\,$. Therefore, \thinspace$C_{\theta}^{Q,C}\left(  G_{\theta},\mathcal{L}%
\right)  =0$ for such models. In this subsection, we show that such a
phenomena can occur only at the surface of the space of the quantum operations.

\begin{theorem}
Suppose (CM.1-2) \ and (CE) holds. Moreover, suppose $\exists\varepsilon>0$
s.t., $\Lambda_{\theta}+\sum_{i=1}^{m}u^{i}\partial_{i}\Lambda_{\theta}$ is
completely positive, for $\forall u:\left\Vert u\right\Vert <\varepsilon$,
$\forall\theta\in\Theta$. Then, $C_{\theta}^{Q,C}\left(  G_{\theta
},\mathcal{L}\right)  \neq0$.
\end{theorem}

\begin{proof}
Define $\Lambda_{\theta_{0},\theta}:=\Lambda_{\theta_{0}}+\sum_{i=1}^{m}%
\frac{\partial\Lambda_{\theta_{0}}}{\partial\theta^{i}}\left(  \theta
^{i}-\theta_{0}^{i}\right)  $, $\mathcal{L}_{\theta_{0}}:=\left\{
\Lambda_{\theta_{0},\theta};\,\sum_{i=1}^{m}\,\left\vert \theta^{i}-\theta
_{0}^{i}\right\vert <\varepsilon\right\}  $. Suppose $\mathcal{E}_{\theta
_{0},n}:=\{\rho_{\theta_{0}}^{n},\left\{  \Xi_{\kappa,\theta_{0}}^{n}\right\}
_{\kappa=1}^{n},\mathsf{M}_{\theta_{0}}^{n},T_{\theta_{0},n}\}$ satisfies (CE)
and (\ref{locally-unbiased}) at $\theta=\theta_{0}$ as an estimator of
$\mathcal{L}:=\left\{  \Lambda_{\theta}\right\}  _{\theta\in\Theta}$ .
\begin{align*}
&  \partial_{i}\left[  \int T_{\theta_{0},n}^{j}\left(  \omega_{n}\right)
\mathrm{tr}\left[  \,M_{\theta_{0}}^{n}\left(  \mathrm{d}\,\omega_{n}\right)
\prod_{\kappa=0}^{n}\left\{  \left(  \Lambda_{\theta_{0}}+\sum_{i=1}^{m}%
\frac{\partial\Lambda_{\theta_{0}}}{\partial\theta^{i}}\left(  \theta
^{i}-\theta_{0}^{i}\right)  \otimes\mathbf{I}\right)  \circ\Xi_{\kappa
,\theta_{0}}^{n}\right\}  \rho_{\theta_{0}}^{n}\right]  \right]
_{\theta=\theta_{0}}\\
&  =\partial_{i}\left[  \int T_{\theta_{0},n}^{j}\left(  \omega_{n}\right)
\mathrm{tr}\left[  \,M_{\theta_{0}}^{n}\left(  \mathrm{d}\,\omega_{n}\right)
\prod_{\kappa=0}^{n}\left\{  \left(  \Lambda_{\theta}\otimes\mathbf{I}\right)
\circ\Xi_{\kappa,\theta_{0}}^{n}\right\}  \rho_{\theta_{0}}^{n}\right]
\right]  _{\theta=\theta_{0}}=\delta_{i}^{j}.
\end{align*}%
\vskip-\lastskip
Here, the first identity holds since (CE) implies an analogue of
Lemma$\,$\ref{lem:est-cont}. Therefore, \ since $\Lambda_{\theta_{0}%
,\theta_{0}}=\Lambda_{\theta_{0}}$, $\mathcal{E}_{\theta_{0},n}$ satisfies
(\ref{locally-unbiased}) at $\theta=\theta_{0}$ as an estimator of
$\mathcal{L}_{\theta_{0}}$, also. Moreover, the variance of $\mathcal{E}%
_{\theta_{0},n}$ as an estimator of $\mathcal{L}_{\theta_{0}}$ and
$\mathcal{L}$ coincide at $\theta=\theta_{0}$. Therefore, it suffices to show
the statement for the quantum operation model $\mathcal{L}_{\theta_{0}}$. From
here, we follows the same line of argument as Hayashi\thinspace(2003) and
Zhengfeng Ji, et. al.\thinspace(2006).

Let $\theta_{x}$ be the $x$-th extreme point of the convex region $\sum
_{i=1}^{m}\,\left\vert \theta^{i}-\theta_{0}^{i}\right\vert \leq\varepsilon$.
Then, there is $p_{\theta_{0},\theta}\left(  x\right)  $ such that
$\Lambda_{\theta_{0},\theta}=\sum_{x=1}^{2^{m}}p_{\theta_{0},\theta}\left(
x\right)  \Lambda_{\theta_{0},\theta_{x}}$, $\sum_{x=1}^{2^{m}}p_{\theta
_{0},\theta}\left(  x\right)  =1$, and $p_{\theta_{0},\theta}\left(  x\right)
$ is linear in $\theta$.\quad Consider the family of multinomial probability
distributions $\{$ $p_{\theta_{0},\theta}\left(  \cdot\right)  \}_{\theta
\in\Theta}$. The key observation is that $\Lambda_{\theta_{0},\theta}$ is
equivalent to random application of $\Lambda_{\theta_{0},\theta_{x}}$, where
$x$ is sampled according to $p_{\theta_{0},\theta}\left(  \cdot\right)  $.
Following this observation, given a locally unbiased estimator $\mathcal{E}%
_{\theta_{0},n}$ of $\mathcal{L}_{\theta_{0}}$ at $\theta=\theta_{0}$, one can
compose a locally unbiased estimator of the statistical model $\left\{
p_{\theta_{0},\theta}\left(  \cdot\right)  \right\}  _{\theta\in\Theta}$ :
Prepare a quantum state $\rho_{\theta_{0}}^{n}$, apply a sequence of quantum
operations $\prod_{\kappa=1}^{n}\left\{  \left(  \Lambda_{\theta_{0}%
,\theta_{x_{\kappa}}}\otimes\mathbf{I}\right)  \circ\Xi_{\kappa,\theta_{0}%
}^{n}\right\}  $, where $x_{\kappa}\sim p_{\theta_{0},\theta}$ ($\kappa
=1$,$\cdots$,$n$), measure the output state by $M_{\theta_{0}}^{n}$, and
compute $T_{\theta_{0},n}$. Therefore, due to classical estimation theory, the
variance of $\mathcal{E}_{\theta_{0},n}$ is $O\left(  1/n\right)  $.
\end{proof}

\vskip3mm

\noindent\textbf{BIBLIOGRAPHY} \vskip3mm

\noindent DasGupta, A., "Asymptotic Theory of Statistics and Probability,"
Springer Verlag, New York (2008)\vskip1mm

\noindent Fujiwara, A., "Estimation of a generalized amplitude-damping
channel," Phys. Rev. A 70, 012317 (2004).\vskip1mm

\noindent Gill, R. and Guta, M., "An Introduction to Quantum Tomography,"
http://lanl.arxiv.org/abs/quant-ph/0303020 (2006).\vskip1mm

\noindent Gill, R. and Massar, S., "State estimation for large ensembles,"
Phys.Rev. A61, 042312, (2002).\vskip1mm

\noindent Hayashi, M., "Quantum estimation and the quantum central limit
theorem", Bulletin of Mathematical Society of Japan, Sugaku, Vol. 55, No. 4,
368--391 (2003, in Japanese)\vskip1mm

\noindent Hayashi, M., eds., Asymptotic Theory of Quantum Statistical
Inference: Selected Papers, (WorldScientific, Singapore, 2005).\vskip1mm

\noindent Hayashi, M. and Matsumoto, K., \textquotedblleft Statistical model
with measurement degree of freedom and quantum physics,\textquotedblright%
\ RIMS Kokyuroku Kyoto University, No. 1055, 96--110 (1998) (In Japanese). (
English translation appeared as Chap. 13 of Hayashi 2005.)\vskip1mm

\noindent Hayashi, M. and Matsumoto, K., \textquotedblleft Asymptotic
performance of optimal state estimation in quantum two level
system,\textquotedblright\ http://lanl.arxiv.org/abs/quant-ph/0411073 (2004).\vskip1mm

\noindent Hayashi, M., A new proof of the main theorem of Fujiwara\thinspace
(2004), private communication (2003)\vskip1mm

\noindent Holevo, "Probabilistic and Statistical Aspects of Quantum Theory,"
(North-Holland, Amsterdam, 1982) (in Russian, 1980).\vskip1mm

\noindent Matsumoto,K., seminar note (1999)\vskip1mm

\noindent Matsumoto,K., " A new approach to the Cramer-Rao-type bound of the
pure-state model," J. Phys. A, 35-13, 3111-3123 (2002) .\vskip1mm

\noindent Matsumoto, K., "Self-teleportation and its application to LOCC
estimation and other tasks", \newline
http://lanl.arxiv.org/abs/quant-ph/0709.3250 (2007).\vskip1mm

\noindent Nagaoka, H., \textquotedblleft On Fisher Information of Quantum
Statistical Models, SITA'87, 19-21, Nov., (1987) (in Japanese).\vskip1mm

\noindent Nagaoka, H., A New Approach to Cramer-Rao Bounds for Quantum State
Estimation, IEICE Technical Report, IT89-42, 9-14(1989). \vskip1mm

\noindent Ozawa, M., \textquotedblleft Conditional probability and a
posteriori states in quantum mechanics,\textquotedblright\ Publ. Res. Inst.
Math. Sci. Kyoto Univ., 21: 279--295 (1985).\vskip1mm

\noindent Schwabik, S. and \ Guoju, Y., "Topics in Banach Space Integration,"
\ (World Scientific, 2005).\vskip1mm

\noindent Zhengfeng Ji, et. al., "Parameter estimation of quantum channels,"
http://lanl.arxiv.org/abs/quant-ph/0610060 (2006). \vskip1mm

\vskip4mm

\end{document}